\documentclass[a4paper,12pt,reqno]{amsart}
\usepackage[T1]{fontenc} 
\usepackage[utf8]{inputenc} 
\usepackage[english]{babel} 
\usepackage{url} 
\usepackage{amsmath}
\usepackage{amssymb}
\usepackage{amsthm}
\usepackage{mathrsfs}
\usepackage{mathtools}
\usepackage[italian]{varioref}
\usepackage{hyperref}
\usepackage{typehtml}
\usepackage{multicol}
\usepackage{enumerate}
\usepackage{color}
\usepackage{mathrsfs}
\usepackage{comment}
\usepackage{mathabx}
\usepackage{tikz-cd}
\usepackage{adjustbox}
\usepackage{bbm}

\usepackage{graphicx}

\usepackage{color}
\newcommand{\grad}{\nabla}
\newcommand{\lap}{\triangle}

\newcommand{\dde}[2]{{\dfrac {\partial {#1}}{\partial {#2}}}}

\newcommand{\R}{\mathbb R}

\newcommand{\N}{\mathbb N}
\newcommand{\C}{\mathbb C}
\newcommand{\e}{\mathrm{e}}                
\newcommand{\de}{\mathop{}\!\mathrm{d}}  
\newcommand{\pa}{\mathop{}\!\partial} 
\newcommand{\opde}[1]{{\dfrac {\mathrm{d} \phantom{#1}}{\mathrm{d} {#1}}}}

\newcommand{\eps}{\varepsilon}

\newcommand{\fcr}{\mathcal X}
\newcommand{\Cal}[1]{{\mathcal {#1}}}
\newcommand{\ms}[1]{{\mathscr {#1}}}
\newcommand{\mf}[1]{{ {#1}}}

\newtheorem{theorem}{Theorem}[section]
\newtheorem{proposition}[theorem]{Proposition}
\newtheorem{lemma}[theorem]{Lemma}
\newtheorem{corollary}[theorem]{Corollary}
\theoremstyle{definition}

\newtheorem{remark}[theorem]{Remark}
\numberwithin{equation}{section}
\numberwithin{theorem}{section}

\begin{document}

\pagestyle{plain}

\title[Microcanonical phase transitions]
{Microcanonical phase transitions for the vortex system}

\author[D.\ Benedetto]{Dario Benedetto
      }
\address{Dario Benedetto \hfill\break \indent
        Dipartimento di Matematica, Universit\`a di Roma `La Sapienza'
        \hfill\break \indent
        P.le Aldo Moro 2, 00185 Roma, Italy,
         \hfill\break\indent
        INdAM - Istituto Nazionale di alta Matematica GNFM, Roma, Italy
      }
      \email{benedetto@mat.uniroma1.it}

 \author[E.\ Caglioti]{Emanuele Caglioti}
 \address{Emanuele Caglioti \hfill\break \indent
   Dipartimento di Matematica, Universit\`a di Roma `La Sapienza'
   \hfill\break \indent
   P.le Aldo Moro 2, 00185 Roma, Italy,
   \\
    \hfill\break\indent
   INdAM - Istituto Nazionale di alta Matematica GNFM, Roma, Italy}
 \email{caglioti@mat.uniroma1.it}

 \author[M.\ Nolasco]{Margherita Nolasco}
 \address{Margherita Nolasco \hfill\break \indent
   Dipartimento di Ingegneria e Scienze dell’informazione e Matematica,
    \hfill\break \indent
   Universit\`a dell’Aquila
   Via Vetoio, Loc. Coppito, 67100
   L'Aquila,  Italy,
 \hfill\break \indent
INdAM - Istituto Nazionale di alta Matematica GNAMPA, Roma, Italy}
\email{nolasco@univaq.it}

\begin{abstract}
  
  We consider the Microcanonical Variational Principle for the vortex
  system in a bounded domain. In particular
  we are interested in the thermodynamic
  properties of the system in domains of second kind,
  i.e. for which the equivalence of ensembles does not hold.
    
  For connected domains close to the union of disconnected disks
  (dumbbell domains), 
  we show that the
  system may exhibit an arbitrary number of fist-order phase transitions,
  while the entropy 
  is convex for large energy.

\end{abstract}



\maketitle
\pagestyle{headings}

\section{Introduction}
Starting form the pioneering paper by Onsager \cite{O}, statistical
mechanics of point vortices has been widely studied, in particular in
the mean field limit \cite{CLMP1,CLMP2,K,KL,LP,MJ}.  In the case of a
fluid confined in a bounded simply connected set $\Lambda$ in $\R^2$,
the structure of the mean field variational principle and of the
related mean field equation presents many interesting features.  First
of all, as suggested by Onsager, negative temperatures are allowed
and, according to the Microcanonical Variational Principle (MVP) the
entropy  $\ms {S}$ decreases in the energy $E$, if $E$ is 
sufficiently large, and the vorticity density
converges to a delta function as $E\rightarrow \infty.$

Moreover, depending on the shape of the domain, the equivalence of
ensemble can be broken. If the domain is a disk or a
domain close to a disk then the entropy $\ms S(E)$ is a concave function of
the energy and the canonical variational principle and the
microcanonical variational principle are equivalent.  Conversely, for
sufficiently long and thin domains, the two ensemble are not anymore
equivalent: the entropy is not a concave function of the energy and
the mean field equation associated to the canonical variational principle (CVP)
has not a unique solution for some values
of the inverse temperature $\beta<-8\pi$ \cite{CLMP2}.
These two kinds of domains are called
respectively domains of first and second type.

Recently the behavior of the MVP  in domain of second kind
has been
carefully analyzed \cite{BDM,B},  and 
 the natural question if the
entropy is definitively convex for large
values of the energy has been posed.
In \cite{BJLY} 
it has been proved that if  a second kind domain is convex
then  the  branch of solutions of the MVP is regular and
the entropy is definitely convex.

In the present paper we construct non convex connected second kind
domains for which the entropy is definitely convex.
Moreover the branch is not regular and 
the entropy exhibits first order phase transitions, i.e.  jumps
of the derivative w.r.t. the energy.

To construct these solutions, we first solve the MVP for disconnected domains.
In particular we consider domains union of disks or of slightly
deformed disks. In the case of more than three disks,
the entropy has a first order phase transition at low energy.
Considering $N$  suitably deformed disks we also construct
disconnected domains which exhibit $N$ first order phase
transitions at  high energies.
Then, connecting the (deformed) disks with thin channels,
we obtain connected domains (dumbbell domains), for which 
the branches of solutions and the phase transitions of the entropy
are preserved.

\vskip.3cm

Now we give some
definitions and recall some results.
Let $\Lambda\subset \R^2$ be a bounded open  set with smooth boundary.
For a distribution probability with density $\rho\in L^1(\Lambda)$, 
the entropy and the energy functionals are defined respectively as
\begin{equation}
  \label{eq:S}
  \Cal S(\rho) = -\int_{\Lambda} \rho \ln \rho
\end{equation}
\begin{equation}
  \label{eq:E}
  \Cal E(\rho) = \frac 12 \int_{\Lambda} \rho \Psi =
  \frac 12 \int |\grad \Psi|^2,
\end{equation}
where the stream function  $\Psi$  solves
\begin{equation}
  \label{eq:psi}
  -\lap \Psi = \rho, \ \ \text{ with } \ \left.
    \Psi\right|_{\pa \Lambda} = 0.
\end{equation}
Denoting with $G(x,y)$
the Green function of the Poisson problem in $\Lambda$
with homogeneous Dirichlet boundary conditions,
\begin{equation}
  \label{eq:green1}
  \Psi(x)  = \int_{\Lambda} G(x,y) \rho(y) \de y.
\end{equation}
We write
\begin{equation}
  \label{eq:green}
  G(x, y ) = - \frac 12 \ln |x - y| + \gamma(x, y)
\end{equation}
where $\gamma$ is the regular part of $G$.

We also define the free-energy functional
\begin{equation}
  \label{eq:F}
\Cal F(\beta, \rho) =   \Cal E(\rho)-\frac{1}{\beta}\Cal S(\rho)
\end{equation}
where $\beta\in \R$ is the inverse of the temperature. Note that in 
the case of the statistical mechanics of vortices, $\beta$
can be also negative.
\vskip.3cm


The MVP - Microcanonical Variational Principle
is the problem of finding the distribution $\rho$
which maximizes $\Cal S(\rho)$ among the probability distribution with fixed
energy $E>0$:
\begin{equation}
  \label{eq:mvp}
  \begin{aligned}
    &\ms {S}(E) = \sup_{\rho\in P_E} \Cal S(\rho),\\
    &\text{where} \
  P_E = \left\{ \rho \in L^1(\Lambda)| \, \rho \ge 0, \ \int_\Lambda
    \rho = 1, \
    \Cal E(\rho) = E\right\}.
  \end{aligned}
\end{equation}
When needed, we specify
with $\Cal S_\Lambda$, $\ms {S}_\Lambda$, $P_{\Lambda,E}$
etc.  the dependence on the set  $\Lambda$.

We indicate with $\rho_m = 1/|\Lambda|$
the uniform probability distribution,
with  energy $E_m$.
We recall the following results from  \cite{CLMP2}.

\begin{theorem}[Microcanonical Variational Principle]
  \label{propo:mvp}
  Let $\Lambda$ be a bounded open connected set.
  For any $E>0$, $-\infty < \ms {S}(E) \le \ln |\Lambda|$
  and there exists $\rho \in P_E$ such
  that $\ms {S}(E) = \Cal S(\rho)$. Moreover $\rho>0$ and there exists
  an unique $\beta\in \R$ such that the stream function
  satisfies the mean-field equation (MFE) 
  \begin{equation}
    \label{eq:mfe}
    -\lap \Psi = \rho = \frac 1{Z} \e^{-\beta \Psi},\ \
    Z = \int_\Lambda \e^{-\beta \Psi},
  \end{equation}
  with homogeneous Dirichlet boundary conditions.
    
  The maximum of  $\ms {S}(E)$ is  reached  at  $\ms S_m
  = \Cal S(\rho_m) = \ln |\Lambda|$.
  $\ms {S}(E)$ is  
  continuous, 
  strictly increasing for $E< E_m$,
  strictly decreasing for $E>E_m$
  and $\ms {S}(E)\to -\infty$  for $E\to 0$ and $E\to +\infty$.
\end{theorem}


Notice that the MFE \eqref{eq:mfe} is also the Euler equation associated
to the CVP 
- Canonical Variational Principle, introduced in \cite{CLMP1, K}, that is
the extremal problems for  the free-energy functional $ \Cal F(\beta,\rho) $:
\begin{equation}
  \label{eq:cvp}
  \begin{aligned}
    &\ms {F}(\beta) =     \begin{cases}   \inf_{\rho\in P} \Cal F(\beta,\rho) & \text{if} \, \, \beta >0  \\  \sup_{\rho\in P} \Cal F(\beta,\rho)  &  \text{if} \, \, \beta < 0   \end{cases}  \qquad \\
    &\text{where} \ 
  P= \left\{ \rho \in L^1(\Lambda)| \, \rho \ge 0, \
    \int_\Lambda\rho = 1\right\}
  \end{aligned}
\end{equation}
Therefore we have two different variational principles, CVP and MVP,
which have the same Euler equation, but clearly this fact is
not sufficient for the equivalence of the two sets of solutions ({\it
  equivalence of ensembles}).

We now summarize some known results on the MFE, 
the CVP and the equivalence of ensembles,
which will be useful in the sequel.
Fixed $\beta \ge 0$, the MFE has a unique solution.
The case $\beta <0$ is more delicate to deal with.
For $\beta \in (- 8 \pi, 0]$ and simply connected
domain, the uniqueness and the regularity of the branch of solutions
of the MFE \eqref{eq:mfe} was proved in \cite{S}, and then extended to
multiply connected sets in \cite{BaLin}.  Moreover, from \cite{CCLin}
(see also \cite{BaLin}), we also know that for $\beta = - 8 \pi$ 
the MFE has at most one solution for connected domain.

Let $ E(\beta) \coloneqq \Cal E(\rho_\beta) $ be the energy of the
(unique) solution of the MFE at inverse temperature $\beta > - 8 \pi$
and define
 $$
 (0,+\infty] \ni E_{-8\pi} = \lim_{\beta \to (-8\pi)^+} E(\beta).$$ As
 in \cite{CLMP2}, we say that a bounded open 
 connected  $\Lambda$ is a {\it first kind domain}
 if $E_{-8\pi}= + \infty$ and is a {\it second kind domain} if
 $E_{-8\pi} < + \infty$.





  \begin{theorem}[CVP - Equivalence of ensembles
    \cite{CLMP1,CLMP2,BaLin}]
\label{thm:CVP}
 Let $\Lambda$ be a bounded open connected set with smooth boundary. 
  \begin{enumerate}[i)]
  
  \item $\ms F(\beta)$ is finite iff $\beta \ge -8\pi$.  For
    $\beta> -8\pi$ the CVP has a unique solution $\rho_\beta$, and the
    stream function $\Psi_{\beta} $ solves the MFE \eqref{eq:mfe}.
    Moreover the set of solutions
    $\{ \Psi_{\beta} \, : \, \beta > - 8 \pi \}$ is a regular branch.

  \item $\beta \mapsto \beta \ms F (\beta)$
    is  a strictly concave
    increasing and continuously differentiable  function,
    $(- 8 \pi, + \infty) \ni \beta  \mapsto
    E(\beta) \in (0,E_{-8\pi}) $ is decreasing and  continuously
    differentiable,
    $\ms {S}(E) = \inf_{\beta} (\beta E  - \beta \ms F (\beta)  )$
    is a smooth concave function of $E\in (0,E_{-8\pi})$ and 
    $\ms {S}'(E) = \beta$.

  \item If $E \in (0, E_{-8\pi})$ and $\rho_{E}$ is solution of the
    MVP then $\rho_{\beta} = \rho_{E(\beta)}$ uniquely solves the CVP and 
    MFE.
  \end{enumerate}
\end{theorem}


In view of \emph{ii)} of the theorem, if the  domain is of first kind
$\ms {S}(E)$
is a strictly concave function of the energy for any $E>0$.
Differently, it is possible to show that, 
if the domain is of  second kind, for some range of
values of the energy larger than
${E_{-8\pi}}$, the solution of the MVP is not unique,  and the 
entropy cannot be a concave function;
in particular there is no equivalence of the
ensembles (we refer the reader to
\cite{CLMP1,CLMP2} for further details). 

A  characterization of the domains of  second kind  can be summarized as follows (see  \cite{BaLin}) 

\begin{theorem}[Second kind domains]
  \label{teo:IItipo}
  Let $\Lambda$ be a  bounded open connected set of class $C^1$,
  and let $\rho_\beta$ be the unique solution of the mean field equation
      for $\beta> -8\pi$,
  then the following facts are equivalent:
  \begin{enumerate}[i)]
  \item $\Lambda$ is a  second kind domain.
  
  \item $\ms F (-8\pi)$ is attained and the unique branch of
    maximizers $\rho_\beta$ with $\beta > - 8 \pi$ converges uniformly
    to the maximizer for $\beta = - 8\pi$.

  \item The mean field equation \eqref{eq:mfe} has a (unique) solution
    for $\beta = -8\pi$.

  \end{enumerate}
  Conversely, $\Lambda$ is a first kind domain iff the unique branch
  of maximizers $\rho_\beta$ with $\beta > - 8 \pi$ blows up as
  $\beta \to (-8\pi)^+$. In particular $\rho_\beta$ converges weakly
  to the $\delta$-measure in the point $x_0$ which is the unique
  maximum point of $\gamma(x, x)$ in $\Lambda$.
\end{theorem}

Recalling that $\pa_E \ms S = \beta$, phase transitions
occur if $\{\rho_E\}_E$ is not connected
in $E$.  More precisely, the solution of the MVP jumps
between different branches of solutions
of the mean-field equation, or between different sections of branches.
We prove the following results.

\vskip.3cm
\noindent
{\bf Theorem}  (Low energy phase 
transitions)

\noindent
{\it   If $\Lambda$ is union of $N\ge 3$ disks
  with sufficiently close radii, 
  then 
  there exists $E_*$
  such that $\ms S(E)$ has a first order phase transition for  $E=E_*$
  (see Theorem \ref{teo:le}).

  \noindent
  Moreover, there exist connected ``dumbbell domains'' obtained
  by
  joining
  the $N$ disks with thin channels, 
  such that $\ms S(E)$ has a first order phase transition near $E_*$
  (see Theorem \ref{thm:transizione}).
    }

\vskip.3cm

 \noindent
 {\bf Theorem} 
 (High energy phase transitions)

  \noindent
{\it 
  There exist domains $\Lambda$, disjoint unions of $N$ suitable deformed
  disks
  such that 
  $\ms S(E)$ has $N$ first order phase transitions,
  for sufficiently large values of $E$ (see Theorem \ref{teo:he}).

    \noindent
  Moreover, there exist connected dumbbell domains obtained by
  joining the components of $\Lambda$ with thin channels,
  such that $\ms S(E)$
  has $N$ first order phase transitions
    for sufficiently large values of $E$
  (see Theorem \ref{thm:transizione}).
  }

\vskip.3cm

In Section 2 we introduce the MVP for disconnected domains,
and in Section 3 we analyze the particular case of domains made of $N$
disks.
In Section 4 we show the the phase transitions,
for disconnected domains.
In Section 5 we use perturbative arguments to
extend the results of Sections 4 to connected dumbbell domains.

\section{MVP for disconnected sets}
In this section we face the case of disconnected 
domains, in which we reduce to consider the MVP 
restricted to each connected component, where the probability mass 
is in general less than 1.

To this aim, we first extend the variational principle to
the case of distributions of mass $M>0$.
The MVP becomes
\begin{equation}
  \label{eq:mvpM}
  \begin{aligned}
  &\ms {S}(M,E) = \sup_{
    \rho\in P_{M,E}} \Cal S(\rho),\\
  &\text{where} \
  P_{M,E} = \left\{ \rho \in L^1(\Lambda)| \rho \ge 0, \ \int_\Lambda
    \rho = M, \
    \Cal E(\rho) = E\right\}
  \end{aligned}
\end{equation}

Now we show that $\ms {S}(M,E)$ is simply
given in terms of $\ms {S}(E)\coloneqq \ms {S}(1,E)$, and
we compute the derivative of $\ms S$.
 \begin{lemma}
  \label{lemma:riscalamento}
  \begin{equation}
    \label{eq:sm} 
    \ms {S}(M,E) = M \ms {S}(E/M^2) - M \ln M,
  \end{equation}
  and this value is attained in $\rho\in P_{M,E}$
  which solves
  \begin{equation}
    \label{eq:mfeM}
    \rho = -\lap \Psi = \frac MZ \e^{-\beta \Psi},\ \ \text{ with }
    Z = \int_\Lambda \e^{-\beta\Psi}.
  \end{equation}
  Moreover
  \begin{equation}
    \label{eq:SMZ}
    \ms {S}(M,E) = M \ln (Z/M) + 2\beta E.
  \end{equation}
\end{lemma}
\begin{proof}
Given $\rho\in P_{M,E}$, $\tilde \rho = \rho /M$
is a probability density with 
$\Cal E(\tilde \rho) = \Cal E(\rho)/M^2$, and
$$\Cal S(\rho) = M \Cal S(\tilde \rho) - M \ln M.
$$
We conclude observing that
$\rho \in P_{M,E}$ iff $\tilde \rho \in P_{1,E/M^2}$.
Eq. \eqref{eq:SMZ} follows from the definition
of $S(\rho)$ and the mean-field equation \eqref{eq:mfeM}.
\end{proof}
\begin{lemma}
  \label{lemma:derivate}
  If $\Lambda$ is a connected bounded open set of first kind,
  $\ms S$ depends regularly in $E$ and $M$ and 
  $$
  \begin{aligned}
    &(i)\ \ \ \pa_E \ms {S}(M,E) = \beta\\
    &(ii)\ \ \ \pa_M \ms {S}(M,E) = \ln (Z/M) - 1.
  \end{aligned}
  $$
\end{lemma}
\begin{proof}
Since $\Lambda$ is of first kind,
$\ms S(E)$ is regular  for $E\in (0,+\infty)$. Then 
$$\pa_E \ms {S}(M,E)  = \frac 1M \pa_E \ms S(E/M^2)$$
and $(i)$ follows from the fact that if $\rho$ solves
\eqref{eq:mfeM}, $\rho/M$ satisfy the MFE equation \eqref{eq:mfe}  with inverse
temperature $M\beta$.
By using \eqref{eq:sm}, $(i)$, and \eqref{eq:SMZ}
we obtain $(ii)$.
\end{proof}
We remark that  $ (i)$ and $(ii)$  hold also for domains of  second kind  if $E/M^2 < E_{-8\pi}$.

 \vskip.3cm
 Recalling that 
 for a bounded connected set $\Lambda$
 the solution $\rho$ of equation \eqref{eq:mfeM} is unique
 if $M\beta > -8\pi$,
 we can 
 define
 $$E(M,\beta ) = \Cal E (\rho), \ \ \text{and}
 \ Z(M,\beta) = \int_{\Lambda} \e^{-\beta \Psi}.$$
We set $E(\beta) \coloneqq E(1,\beta)$ and 
$Z(\beta) \coloneqq Z(1,\beta)$.
Since $\rho/M$ and $\Psi/M$ solves  the MFE  \eqref{eq:mfeM} with inverse
temperature $M\beta$, we have
\begin{equation}
  \label{eq:EZ}
  E(M,\beta) = M^2 E(M\beta), \ \ Z(M,\beta) = Z(M\beta).
\end{equation}

\vskip.3cm
Now we  can state our result on the MVP for
a disconnected set.

\begin{theorem}[MVP for disconnected domains]
  \label{teo:mvpdisconnessi}
  Let $\Lambda$ be a set given by 
  $\Lambda = \bigcup_{i=1}^N \Lambda_i$
  where $\Lambda_i$ are open connected bounded sets of first kind, which do not
  intersect.
  Given $E>0$, 
  there exists $\rho \in P_E$ which maximize $\Cal S_\Lambda$.
  The restricted densities $\rho_i \coloneqq 
  \left. \rho \right|_{\Lambda_i}$
  solve the MVP in the sets $P_{\Lambda_i, E_i, M_i}$
  for some $E_i>0$ and $M_i>0$,
  and the entropy $\ms {S}_\Lambda(E)$ satisfies
  \begin{equation}
    \label{eq:SEMi}
    \ms S_\Lambda(E) = \sup_{\begin{subarray}{l}{E_i>0,\,\, \sum E_i= E}
        \\{M_i>0,\, \sum M_i= 1}\end{subarray}} \ 
    \sum_{i=1}^N \ms S_{\Lambda_i} (M_i, E_i).
  \end{equation}
  
  Moreover, $\rho >0$ and there exists
  $\beta\in \R$ such that the stream function $\Psi$ 
  satisfies the MFE  \eqref{eq:mfe} in $\Lambda$;
  the restricted density $\rho_i $ ($i=1,\dots N$) satisfies 
  \begin{equation}
    \label{eq:mfei}
    \rho_i = -\lap \Psi_i =
    \frac {M_i}{Z_i} \e^{-\beta \Psi_i} \quad \text{in} \, \, \Lambda_i ,
    \ \text{ with }
    Z_i = \int_{\Lambda_i}\e^{-\beta \Psi_i} , 
  \end{equation}
  and $Z_i/M_i = Z = \int_{\Lambda}\e^{-\beta \Psi} $,   for all $i=1,\dots N$. 
\end{theorem}

\begin{proof}
Easily  adapting the proof of
Theorem \eqref{propo:mvp} to $\Lambda$, it can be proved
that there exists $\rho \in P_E$ which maximize $\Cal S_\Lambda$, and that  
the support of $\rho$ is the whole set $\Lambda$. This allow us to
prove that the MVP in $P_{\Lambda_i,E_i,M_i}$ is solved by the restricted
density $\rho_i$ where $E_i$ and $M_i$ are the energy and the mass of $\rho_i$.
As a consequence, $\ms S_\Lambda(E)$ is given by \eqref{eq:SEMi}.
Again from 
Theorem
\ref{propo:mvp} and lemma \ref{lemma:riscalamento}, each $\rho_i$ solves
 \eqref{eq:mfeM} for some $\beta_i$. 
We conclude the proof by showing that $\beta_i$ 
and $Z_i/M_i$ do not depend on  $i$, 
which implies that the partition function in $\Lambda$ is 
$$Z = \sum_{j=1}^n  Z_j = \sum_{j=1}^n  \frac {Z_j}{M_j} M_j = \frac {Z_i}{M_i}, \qquad \forall i=1, \dots N$$
 and  $\rho$ solves the MFE \eqref{eq:mfe} with $\beta = \beta_i$.
 
Indeed, the maximum problem \eqref{eq:SEMi}
can be solved by looking for the critical points of
\begin{equation}
\label{eq:Lagrange_mult}
\sum_{i=1}^N \ms S_{\Lambda_i} (M_i, E_i ) -(\alpha - 1) 
\left( \sum_{i=1}^N M_i - 1\right)
- \beta \left( \sum_{i=1}^N E_i - E\right)
\end{equation}
where $\alpha$ and $\beta$ are Lagrange multipliers.
By deriving w.r.t.\  $E_i$ and $M_i$  the function \eqref{eq:Lagrange_mult}, using  Lemma \ref{lemma:derivate}, for any $i=1, \dots N$,
we get 
$$
\beta_i  = \pa_{E_i} \ms S_{\Lambda_i}(M_i,E_i) = \beta,\ \ \ 
\ln (Z_i/M_i)   = \pa_{M_i} \ms S_{\Lambda_i}(M_i,E_i) +1 = \alpha.
$$
\end{proof}

We remark that if $\rho$ is a probability  density with
energy $E$ and solves the
MFE \eqref{eq:mfe} with inverse temperature 
$\beta$, 
then the restricted density  $\rho_i$ with energy $E_i$
solves the equation \eqref{eq:mfei},
with $M_i = \int_{\Lambda_i} \rho$ and $Z_i = M_i Z$.
Conversely, if $\rho_i$ satisfies  \eqref{eq:mfei}
for some $\beta$, $\rho$ solves the MFE \eqref{eq:mfe},
if the ratios  ratios  $Z_i/M_i$  do not  depend on  $i$.

As for the case of a bounded  connected set, we indicate with 
$E_m$ the value of the energy for which it is achieved the maximum
of the entropy $\ms S_m = \log |\Lambda|$, corresponding to
$\beta = 0$ in \eqref{eq:mfe}.
In the sequel we focus on energies larger than $E_m$,
which is the case of  negative temperature $\beta <0$.
As we show in the next section, even if the domains
$\Lambda_i$ are of first kind,
the MVP for disconnected 
set $\Lambda$
may have solutions with inverse temperature $\beta \leq - 8\pi$.

\vskip.3cm
In order to solves the MVP for disconnected
domains, we need to find the solutions of MFE
and then compare the value of the entropy of the solutions, 
as in \eqref{eq:SEMi}.
We now show how to construct branches of solutions.
To simplify the notation, we express the thermodynamic quantities
for a first kind domain $\Lambda_i$ at inverse negative
temperature $\beta \in ( -8\pi,0]$, as a function of the
parameter
$\mu = - \beta/(8\pi)\in [0,1)$: if $\rho$ is the solution
of the MFE with mass one in $\Lambda_i$
\begin{equation}
  \label{def:ezmu}
  \mf e_i(\mu) = \ms E_{\Lambda_i} (\rho), \ \ \mf z_i(\mu) =
  \int_{\Lambda_i} \e^{-\beta \Psi}.
\end{equation}
so that, for the scaling properties \eqref{eq:EZ},
setting $\mu = -M\beta/(8\pi)$:
\begin{equation}
  \begin{aligned}
    &E_i(M,\beta) = M^2 \mf e_i(\mu) = \frac {(8\pi)^2}{\beta^2} \mu^2 \mf e_i(\mu)\\
    &Z_i(M,\beta) = \mf z_i(\mu)
  \end{aligned}
\end{equation}

\begin{proposition}[Construction of branches of solutions of the MFE]
  \label{propo:costruzione}
  
  In the hypothesis of Thm. \ref{teo:mvpdisconnessi},
  for any fixed
  $\gamma \in (0, \max_i \sup_{\mu\in(0,1)} \frac {\mu}{\mf z_i(\mu)})$,
  let $\mu_i\in (0,1)$ be the solutions of the equations
  $\mu_i = \gamma \mf z_i(\mu_i)$.
  Let $\tilde \rho_i$ be the unique solutions of the MFE
  \eqref{eq:mfe} in $\Lambda_i$
  with inverse temperature $-8\pi \mu_i$.
  Set
  $$\beta = -8\pi \sum_i \mu_i,\ \ M_i = -8\pi \mu_i/\beta, \ \
  \rho_i = M_i \tilde \rho_i.$$
  Then
  $\rho = \sum_i \rho_i(x) \fcr_{\Lambda_i}(x)$
  solves the MFE \eqref{eq:mfe} 
  in $\Lambda$, with energy and entropy given 
  respectively by
  $$
  \Cal E (\rho)  = \frac {(8\pi)^2}{\beta^2} \sum_i \mu_i^2 \mf e_i(\mu_i),
  \ \ \Cal S(\rho)  = -\log \gamma + \log \frac { \beta }{-8\pi} + 2\beta
  \Cal E(\rho).$$
\end{proposition}
\begin{proof} By construction $Z_i/M_i$ does not depends on $i$, so 
the proof is an easy consequence of the previous theorem, the scaling
properties \eqref{eq:EZ}, and \eqref{eq:SMZ} for $M=1$.\end{proof}

In the following section 
we specialize our analysis to the case of 
unions of
disjoint disks.

\section{The landscape of the branch of
  solutions for $N$ disks}
\label{sez:piudischi}

We recall that  for a disk  centered
in $x=0$ and radius $R$, with  area $a=\pi R^2$, the (unique)
solution of the equation \eqref{eq:mfe}
\begin{equation}
  \label{eq:psirho}
  \begin{aligned}
    &\Psi(x) = \frac 2{\beta} \ln \left(
      1+ \frac {\beta}{8\pi}\left(1 - \frac {|x|^2}{R^2}\right)\right),
    \ \ \rho(x) = \frac 1Z \frac 1{ \left(
        1+ \frac {\beta}{8\pi}(1 - \frac {|x|^2}{R^2})\right)^2},\\
    &\text{where} \ 
    Z = \frac {\pi R^2}{1+ \frac {\beta}{8\pi}},
  \end{aligned}
\end{equation}
under the necessary condition $\beta > -8\pi$.
The energy is
$$E= \frac {8\pi}{\beta^2}\left( \frac {\beta}{8\pi}
    - \ln \left( 1 + \frac {\beta}{8\pi}\right) \right).
$$
Then, the quantities defined in  \eqref{def:ezmu} are given by
$$\mf e(\mu) = \frac 1{8\pi} \frac 1{\mu^2} \big( \mu - \log (1-\mu) \big),
\ \ \mf z(\mu) = a/ (1-\mu).$$
Note that the energy $\mf e(\mu)$ does not depends on the area.

\vskip.3cm

Let us specialize  Proposition
\ref{propo:costruzione} to the case
of a domain $\Lambda$, union of $N$ disjoint disks, $D_{i}$,
$i=1,\dots N$, 
of area $a_i$, such that $1=a_1 \ge a_2 \ge \dots \ge a_N > 0$.
Fixed $\gamma$, we have to solve 
\begin{equation}
  \label{eq:mui}
  \mu_i (1-\mu_i) = a_i \gamma \qquad  i=1, \dots N.
\end{equation}
The solutions are
$\mu_i^\pm = \frac 12 \left( 1 \pm \sqrt{1 - 4a_i \gamma}\right)$,
which exist if and only if $\gamma \in [0,1/4]$.

Note that as $\gamma$ goes from $0$ to $1/4$, we have that $\mu_1^-$
increases from $0$ to $1/2$ and $\mu_1^+$ decreases from $1$ to $1/2$,
so that we can parametrize  with $\mu=\mu_1\in [0,1]$
all the other solution $\mu_i$, $i\neq 1$, as follows:
$$\mu_i^\pm = \frac 12\left( 1 \pm \sqrt{1- 4a_i \mu(1-\mu)}\right), \quad 
i = 2,\dots N.$$
Note that  $\mu_i^- \le \mu \le  \mu_i^+$.
Any choice of $\{\mu_i\}_{i\ge 2}$ among the $2^{N-1}$
possibilities  gives
different values $\{M_i,E_i\}_i$ and
different solution $\rho$ of the MFE \eqref{eq:mfe}. In this way, as $\mu$ varies
in $[0,1]$, 
we obtain  different branches
of solution, as we describe below.

\vskip.3cm
\subsection*{$k-$branches}
We fix a subset 
$I^+\subset \{2,\dots N-1\}$,
and its complementary
$I^-=\{2,\dots N-1\}\backslash I^+$:
if $i\in I^+$ we
choose the solution  $\mu_i^+$, if $i \in I^-$ we choose
the solution $\mu_i^-$. 
In this configuration, we put
in $D_{i}$, with $i\in I^+$, a mass greater than $\mu$, the mass in
$D_{1}$, and in $D_{i}$, with $i\in I^-$ a mass less than $\mu$.    

It is convenient to write all the thermodynamic quantities 
as functions of the parameter $\mu \in [0,1]$ as follows
\begin{equation}
  \label{eq:termo}
  \begin{aligned}
  \beta(\mu) & = -8\pi
  \left( \mu + \sum_{i\in I} \mu_i^+ + \sum_{i\in I^-} \mu_i^-
  \right), \\
  Z(\mu) &=
  \frac {-\beta(\mu)a_1}
  {8\pi(1-\mu)\mu} , \\
  E(\mu) &= \frac {8\pi}{\beta^2} \left(\mu^2 \mf e(\mu) +
    \sum_{i\in I} (\mu_i^+)^2\mf e (\mu_i^+) + \sum_{i\in I^-}
    (\mu_i^-)^2 \mf e (\mu_i^-)\right), \\
  S(\mu) &= \ln Z(\mu) + 2 \beta E(\mu).
\end{aligned}
\end{equation}
Note that
if $a_i < 1$ for all $i\ge 2$
all quantities are regular functions of $\mu\in (0,1)$.

For $|I^+| = k$ there are $\binom {N-1}k$
different branches of solutions of  the MFE \eqref{eq:mfe},  
that we call  $k-${\it branches}.
In particular there
exists only one $0-$branch, corresponding to $I^+=\emptyset$.
Note also that if $a_i$
are equal for some $i$, the functions $S(\mu)$,
$E(\mu)$, $\beta(\mu)$ 
can be the same on  different $k-$branches.

\begin{remark}
\label{rem:0_branch}
It is easy to show that  on  the $k-$branches with $k\ge 1$ we 
have  always $\beta < -8\pi$.  
\end{remark}


%


It is useful to extend the definition of domain of
first and second kind to disconnected set.
In the case of union of disconnected disks,
we say that
$\Lambda$ is a \emph{first kind set} if
the $0-$branch entirely lies in the region $\beta > -8\pi$;
if otherwise, 
$\Lambda$ is a \emph{second kind set}. 

\begin{proposition}[Two disks]
  If $n=2$ and $a_2\in (0,1)$, $\Lambda$ 
  is of the first
  kind, $\ms S$ 
  is a  concave function, decreasing for $E \geq E_m$,
  and there is equivalence of ensembles.
\end{proposition}
\begin{proof}
  On the $0-$branch $\beta  > -8\pi$ and 
  $E(\beta)$ is a decreasing function,
  diverging when $\beta \to -8\pi^+$.
  This fact can be used to easily extend 
  the Theorem \ref{thm:CVP}  to this case.
\end{proof}
Note that in the case of two identical disks
the set is of the second kind.
For $N\ge 3$ we can have both first or second kind sets, depending
on the values of $a_i$,  as we show below.

\begin{figure}[ht]
  \hskip-1cm\parbox{15cm}{
    \includegraphics[scale=0.5]{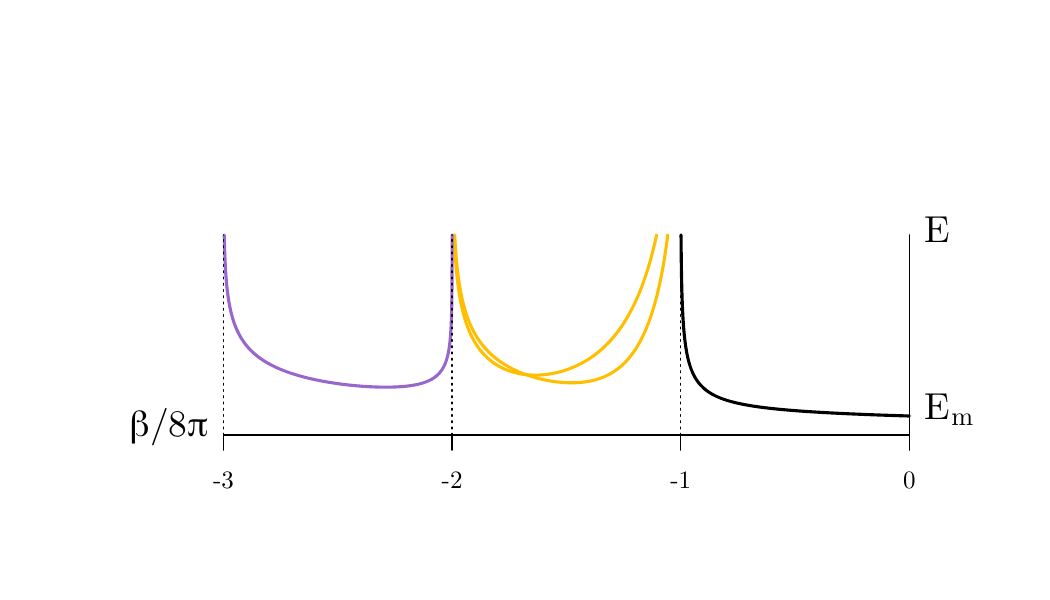}\hskip0.7cm
    \includegraphics[scale=0.5]{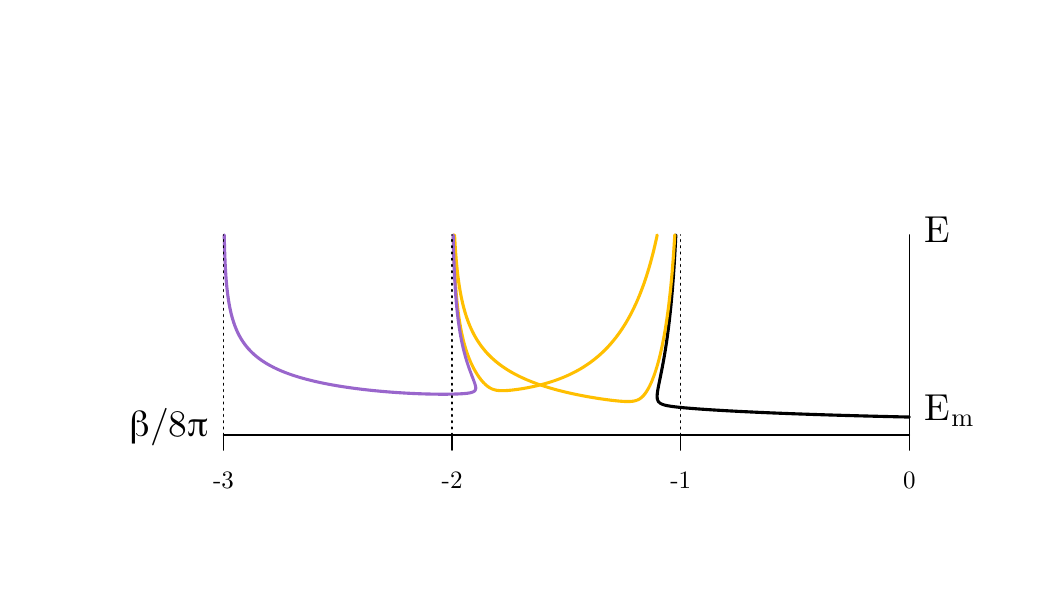}
  }\\ \vskip.5cm
  \hskip-1cm\parbox{15cm}{ \hskip.5cm
    \includegraphics[scale=0.45]{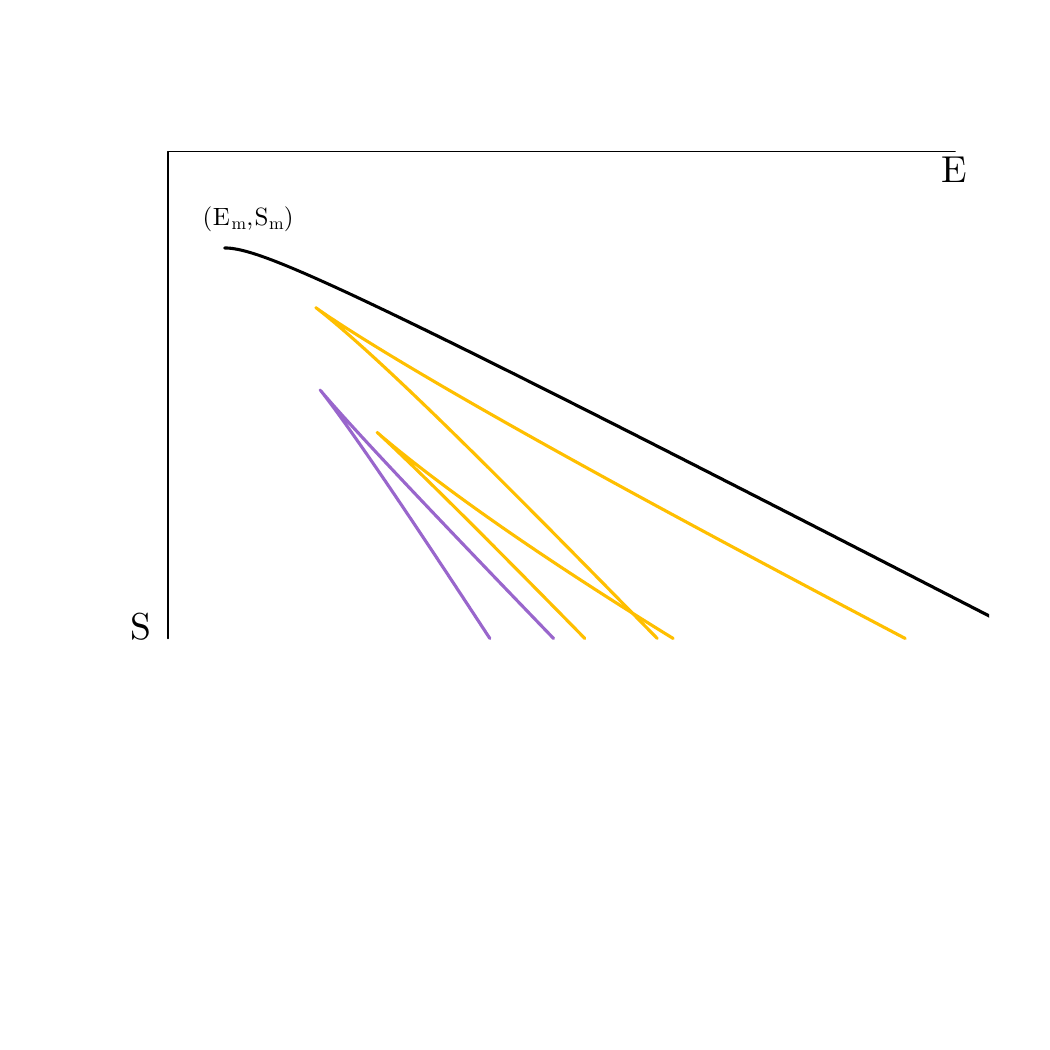}\hskip0.8cm
    \hskip.5cm
    \includegraphics[scale=0.45]{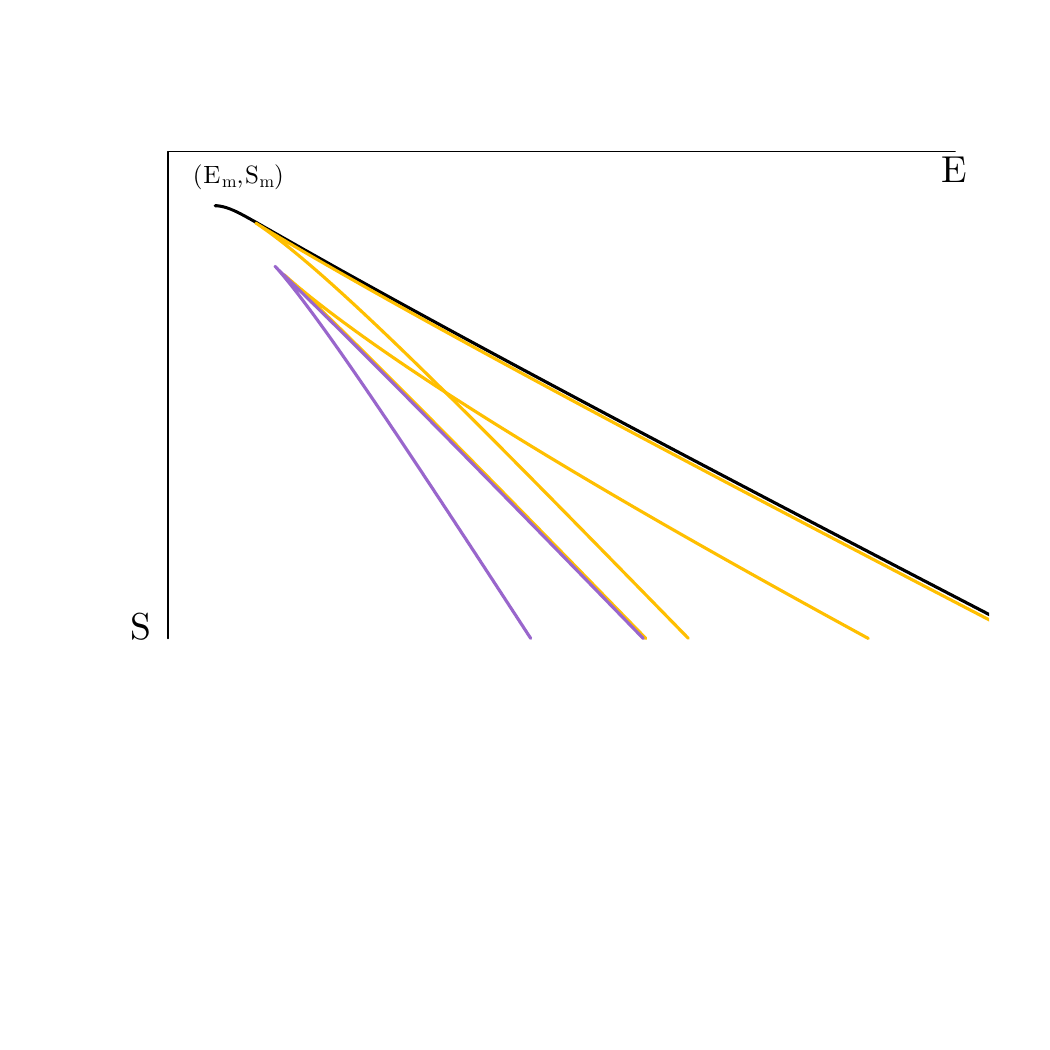}
  }
  \vskip-1cm
  \caption{\small The branches for $n=3$, in the plane $(\beta, E)$
    (graphs above), and in the plane $(E,S)$ (graphs below).
    On the left, $a_2 = 0.6$, $a_3=0.2$; $\Lambda$ is of the
    first kind. On the right, $a_2 = 0.9$, $a_3 = 0.5$; $\Lambda$
    is of the second kind.
  }
  \label{fig:tre12}
\end{figure}

Even if  $S,\,E,\,\beta$ are quite simple  functions  of the variable $\mu$, 
the precise
behavior of the branches 
is not really easy to study  except, 
as we will see later, in the particular  case  $a_i = 1$, 
for all $i=1, \dots N$. So that let us 
give a numerical picture.
If $\mu \to E(\mu)$ is invertible, we write $S(E)$ instead of $S(\mu(E))$,
as well as  if $\mu \to \beta(\mu)$ is invertible we write $E(\beta)$ instead
of $E(\mu(\beta))$.  
In fig. \ref{fig:tre12} we show an example of the different branches
both in the $(\beta,E)$ and  $(E,S)$ planes. 
In particular, on the left, we see that   $E(\beta)$ is
convex on all the $k-$branches, with $k\ge 0$. Moreover, 
since $\pa_E S = \beta$, $\pa_E^2 S = 1/ \pa_\beta E$, on
the increasing part of $E(\beta)$, the entropy $S(E)$ is convex and on
the decreasing part is concave.  On the $0-$branch, $S(E) = \ms S(E)$
is concave then $\Lambda$ is of the first kind.
On the right, 
$\Lambda$ is of the second kind, 
$\ms S(E)$ is concave up to the value $E$
corresponding to the turning point of the $0$-branch
in the plane $(\beta,E)$, in which 
$\beta<-8\pi$ reaches its minimum; then $\ms S(E)$ becomes convex.

Finally, by using \eqref{eq:termo},
it is not difficult to prove that
\begin{theorem}[First kind]
  \label{teo:fistkind}
If $a_i$ are sufficiently small for $i=2, \dots N$ then   $\Lambda$ 
 is of the first kind.
\end{theorem}


\vskip.3cm

The solution $\ms S(E)$ of the  MVP
has a complex behavior when  the $a_i$ are near 1.
In order to show this fact,
we start by  analyzing  the degenerate case,
namely when all the disks have the same area.


\vskip.3cm

\subsection*{$k-$merged-branches}

We consider  
$a_i=1$ for all $i=1,\dots N$.
As $\gamma$ goes form $0$ to  $1/4$,
all  $\mu_i^- $  are equal and  increase  from 
$0$ to $1/2$
and all  $\mu_i^+ $ are equal and  decrease from $1$ to $1/2$,
moreover all the $k$-branches coincide.
It is convenient to re-parametrize the branches as follows.  We set
$\mu_i^- = \mu $ and $\mu_i^+ = 1 - \mu$; extending the values of
the parameter from $\mu \in [0,1/2]$ to $\mu \in [0, 1]$, we get that
the $k$-branch (described by $\mu \in [0,1/2]$ ) ``merge'' with the
$(n-k)$-branch (described by $\mu \in [1/2,1]$) in a unique branch of
solutions, that we call $k-${\it merged-branch},
Namely we get a branch of solutions of the MFE \eqref{eq:mfe}
for which  we choose $(n-k)$ disks with  
$\mu_i = \mu$,  $k$ disks with $\mu_i =  1- \mu$  and  $\mu \in [0,1]$. Clearly 
the $k$ and the $(n-k)$-merged-branch  coincide,
so that  we consider only  $ k\le n/2$.

The thermodynamic quantities in terms of the variable $\mu \in (0,1) $ are
now given by 
\begin{equation}
  \label{eq:termo2}
\begin{aligned}
  \beta(\mu) & = -8\pi \left( (n-k)\mu + k(1-\mu) \right),\\
   Z(\mu) &= \frac {-\beta(\mu)}{8\pi(1-\mu)\mu},\\
   E(\mu) & = \frac {8\pi}{\beta^2} \left( (n-k)\mu^2\mf e(\mu) + k(1-\mu)^2
     \mf e(1-\mu)\right),\\
  S(\mu) &= \ln Z(\mu) + 2 \beta(\mu) E(\mu).
\end{aligned}
\end{equation}
Note that all the merged-branches have a common solution,
reached for $\mu =1/2$, for which
$\beta = \beta_{c} \coloneqq  -4\pi n$, of energy $E_c$ and
entropy $S_c$.

In the $(\beta,E)$ plane, the merged branches are regular,
while all the $k-$branches loose their regularity in
$\beta=\beta_c$. For example, 
the $0-$branch is given by the union of  the restriction
to  the region $\beta \geq \beta_c$ of 
the $0-$merged branch and the 
$1-$merged branch (see  figure \ref{fig:merged}),

\begin{figure}[h]
  \includegraphics[scale=0.45]{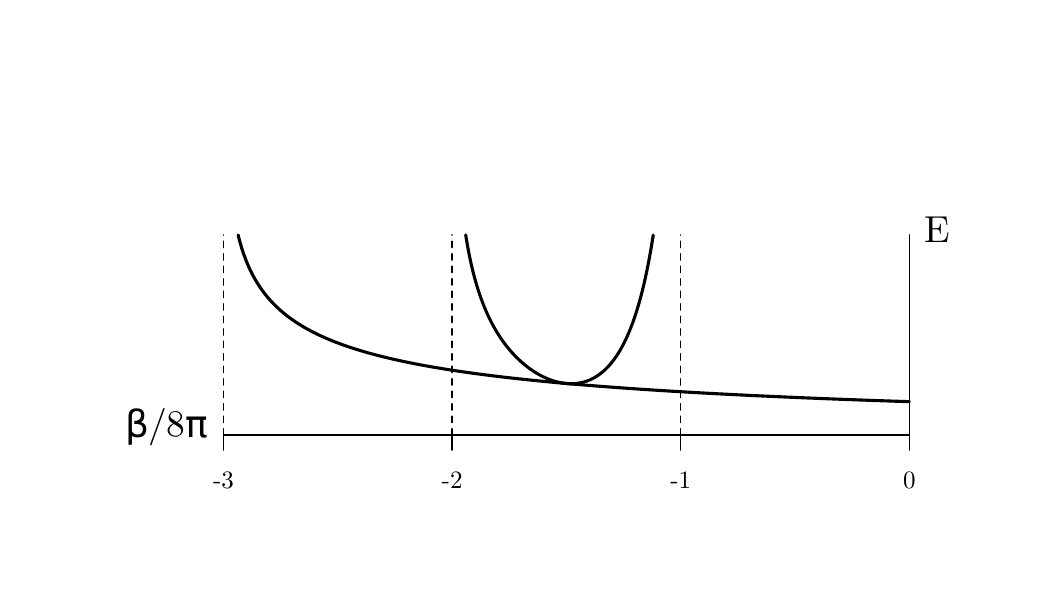}\hskip1cm
  \includegraphics[scale=0.45]{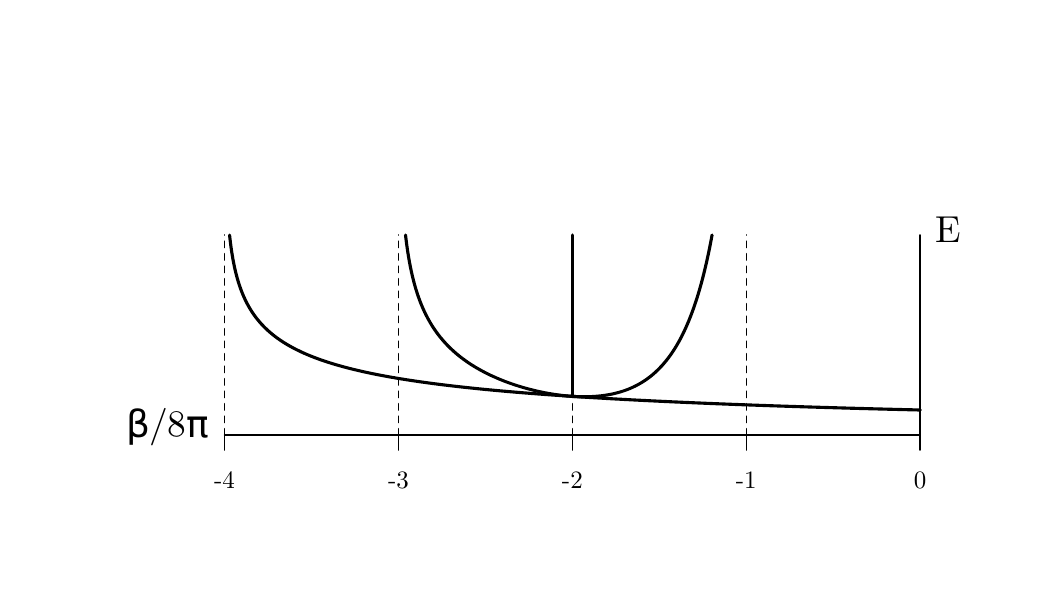}
  \caption{\small The merged-branches for $n=3$ and $n=4$. Any merged-branch
    is union of part of different branches.
    If $n$ is even, all the solutions of the $n/2-$merged-branch
    have the same $\beta = \beta_c$.}
  \label{fig:merged}
\end{figure}

\section{Phase transitions}

In this section we exhibit disconnected sets for which
the solution of the MVP jumps from one point to another
of the branches. As a consequence, $\pa_E \ms S(E)$ is discontinuous.
In this sense, $\ms S(E)$ has a first order phase transition.

We analyze two cases: in the first, $\Lambda$ is an union of
disks, and a phase transition occurs at low energy;
in the second, $\Lambda$ is an union of $N$ deformed disks,
and $\ms S(E)$ has $N$ phase transition
for large value of the energy.

\subsection{Low energy phase transitions}

In this subsection, $\Lambda$ is an union of disks, for which we
constructed the 
branches of solutions  in the previous section.

\begin{theorem}[The branch of $\ms S(E)$]
  \label{teo:0branch}
  If $a_i \in (0,1)$ for $i=2,\dots N$,
  the solution of the MVP lies on the $0-$branch.
\end{theorem}
\begin{proof}
  Note that, in general,  if $\tilde \Lambda = L^{-1}  \Lambda$, with $L>0$,
  then
  \begin{equation}
    \label{eq:area} 
    \ms S_{\tilde \Lambda} (M,E)  -  M \log |\tilde \Lambda| =
    \ms S_{\Lambda} (M,E) - M \log |\Lambda|.
  \end{equation}
  Indeed, 
  if $\rho$ is a probability density
  with support in $\Lambda$,  mass $M$ and energy $E$,
  then $\tilde \rho (x) = L^2 \rho (Lx)$ 
  has the same  mass $M$ and energy $E$ of $\rho$.
  As a consequence, if  $\rho$ maximize  
  $\Cal S_\Lambda$, then
  $\tilde \rho$  maximize 
  $\Cal S_{\tilde \Lambda}$, hence
  $$\ms S_{\tilde \Lambda} (M,E) =
  - \int_{\tilde \Lambda} L^2 \rho(Lx) \log \rho(Lx)
   -\int_{\tilde \Lambda}
   \tilde \rho \log L^2
   = \ms S_{\Lambda} (M,E) - M \log L^2,$$
   where  $L^2 =| \Lambda |  /  |\tilde \Lambda|$.

   We now rewrite the entropy \eqref{eq:SEMi}
   using  the 
   scaling properties \eqref{eq:area}:
   $$\ms S_\Lambda(E) = \sum_{i=1}^N \ms S_{D_{a_i}}(M_i,E_i) = 
   \sum_{i=1}^N \ms S_{D_{a_1}}(M_i,E_i) + \sum_{i=1}^N M_i \log a_i,$$
   hence we get that if $a_i < a_j$ then $M_i \leq M_j$. Indeed, if
   $M_i > M_j$, by exchanging the mass and the energy between the disks
   $D_{a_i}$ and $D_{a_j}$ the entropy increases and we get a
   contradiction since $\ms S_\Lambda(E) $ is the maximum over all the
   possible configurations.  In particular, we conclude that
   $M_i \le M_1$ for all $i\ge 2$, condition fulfilled only on the
   $0-$branch.
 \end{proof}

 \vskip.3cm

 \begin{theorem}[First order phase
   transition of the entropy]
   \label{teo:le}
   If $N\ge 3$ and for all $i =  2, \dots N$, $a_i$ is
   sufficiently close to $1$, then 
   there exists $E_* > E_m$
   such that
   $$\beta^-=\dde {\ms S{\phantom{^\pm}}}{E^-}(E_*) <  \beta^+ = 
   \dde {\ms S{\phantom{^\pm}}}{E^+}(E_*).$$
\end{theorem}
\begin{proof}

  In the case of $a_i=1$ for all $i$,
  in a neighborhood of $(\beta_c,E_c)$, the branches of solutions
  look as in fig. \ref{fig:locale}, on the left,
  (in which are represented 
  the $0$ and the $1-$merged-branches),
  while when $a_i<1$, for $i=2\dots N$ very close to $1$,
  the $0-$branch 
  looks as in the figure on the right.
  The only qualitative difference is that in the case of
  identical disks the $0-$branch becomes singular in
  $\beta = \beta_c$.
  
  We start the proof considering the case of $a_i<1$,
  represented in fig. \ref{fig:locale}, on the right,
  for which  the solution of the MVP lies on the $0$-branch
  (see Theorem \ref{teo:0branch}).

  \begin{figure}[h]
    \includegraphics[scale=0.45]{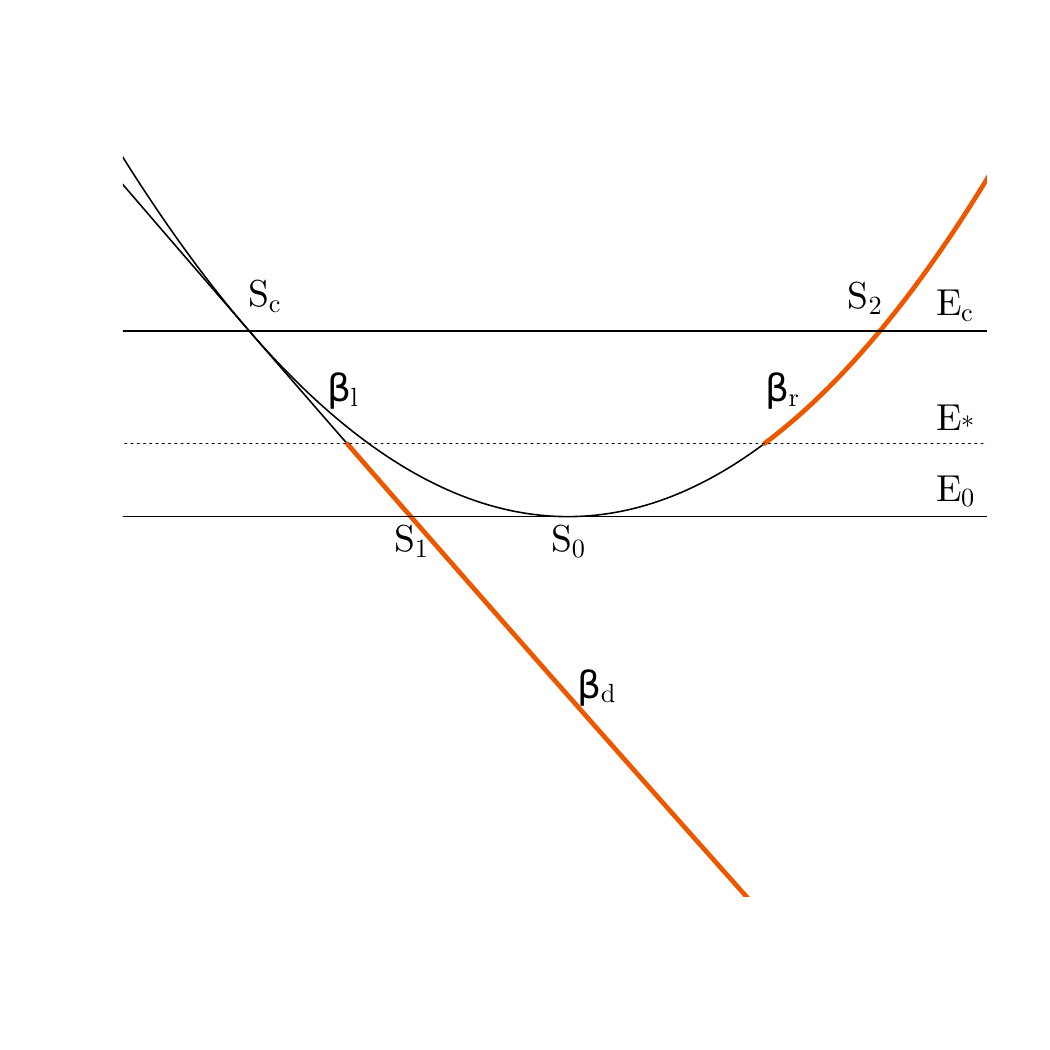}
    \includegraphics[scale=0.45]{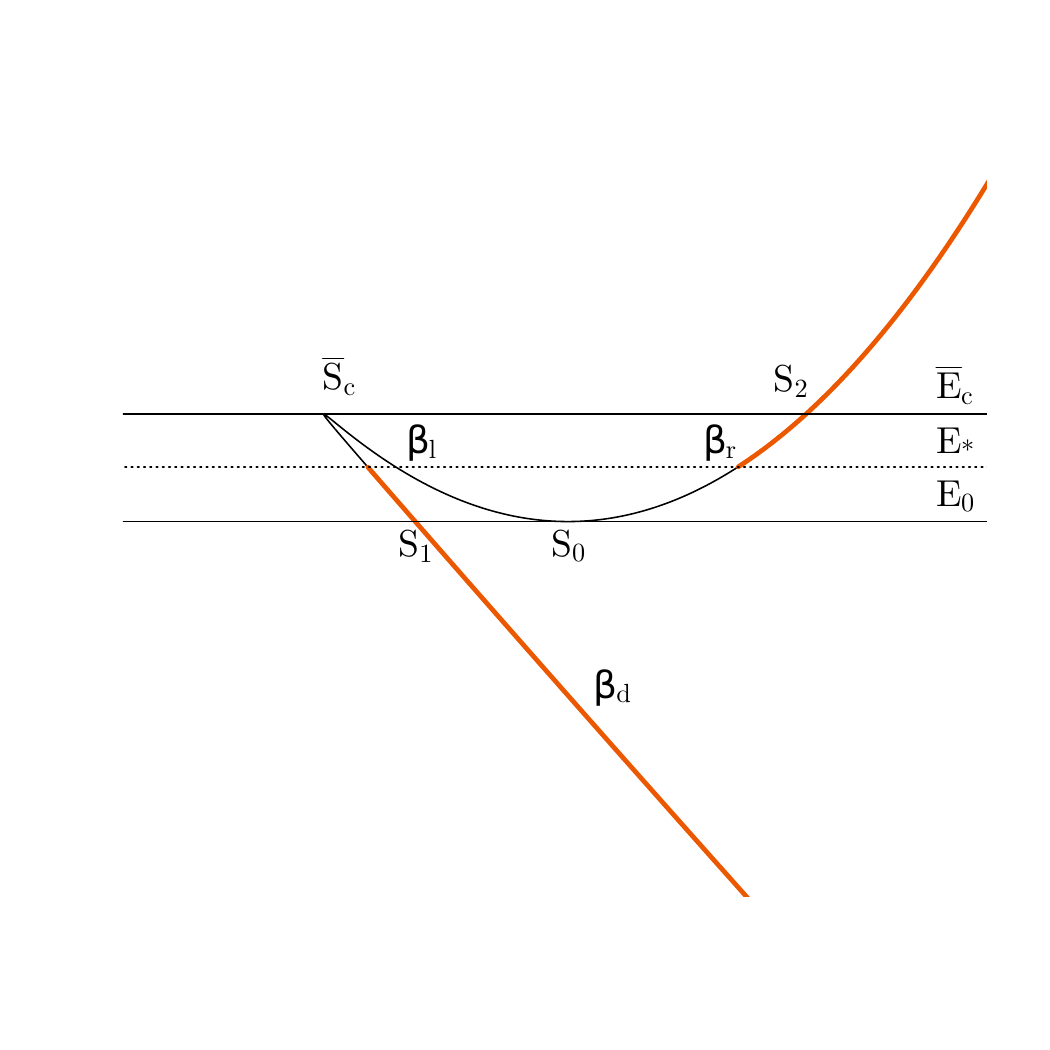}
    \caption{\small On the left, we show
      the $0$ and $1$-merged-branches
      in a neighborhood of $(\beta_c,E_c)$.
      The parameters of the  solution of the MVP lies on the
      red part, and $\beta$ jumps  at $E=E_*$.
      In the right it is represented the $0-$branch in
      the case $a_i<1$, but very
      close to 1
      ($N=3$, $a_2=a_3 = 1 - 3\times 10^{-5}$).}
    \label{fig:locale}
  \end{figure}
  We denote by $\beta_d(E)$ the function $\beta(E)$
  on the $0-$branch, until the local maximum 
  $E={\bar E}_c$, for which the entropy is ${\bar S}_c$.
  On the remaining part of the branch, we call $E_0$ the local minimum,
  reached in $\beta_0$, with entropy $S_0$, and
  we denote by $\beta_l(E)\le \beta_0$ and $\beta_r(E)\ge \beta_0$
  the two values of $\beta(E)$ in function of $E\ge E_0$.
  We denote by $S_d$, $S_l$, $S_r$
  the corresponding values of the entropy, which are regular functions
  of $E$.

  First we note that the argument in Lemma \ref{lemma:derivate}-$(i)$
  applies in general for $\Cal S(\rho)$ when $\rho$ solves
  MFE, if it depends regularly on $E$. Since this is true 
  if $E\neq {\bar E}_c$,
  we have
  $\pa_E S(E) = \beta$, on all  parts of the branch.
  Then,
  for $E\in (E_0,{\bar E}_c)$ we have 
  $$\begin{aligned}
    {\bar S}_c-S_d(E) = \int_{E}^{{\bar E}_c} \beta_d < \int_{E}^{{\bar E}_c}
    \beta_l
    = {\bar S}_c - S_l(E),
    \\
    S_l(E) - S_0 = \int_{E_0}^E \beta_l < \int_{E_0}^E \beta_r = S_r(E) - S_0,
  \end{aligned}
  $$
  hence $S_l(E) < S_d(E)$ and $S_l(E) < S_r(E)$. In particular
  ${\bar S}_c < S_2 \coloneqq S_r(\bar{E}_c)$ and $S_0 < S_1
  \coloneqq S_d(E_0)$.
  Since 
  $$\dde {\phantom{S}}E (S_r(E) - S_d(E)) = \beta_r(E) - \beta_d(E) > 0,$$
  and
  $$S_r(E_0) - S_d(E_0) = S_0 - S_1 < 0,\ \
  S_r({\bar E}_c) - S_d({\bar E}_c) = S_2 - {\bar S}_c > 0,$$
  there exist only one value $E_*\in (E_0,{\bar E}_c)$ such that
  $$
  \begin{array}{ll}
    S_d(E) > S_r(E) & \text{ for } E\in (E_0,E_*)\\
    S_d(E) < S_r(E) & \text{ for } E\in (E_*,{\bar E}_c).
  \end{array}
  $$
  In the case of identical disks,
  the same proof applies, 
  but we have also to prove that the
  the entropy of any other branch is always less than $S_r(E)$
  (Thm. \ref{teo:0branch} does not apply in this case).
  This can be done using the same idea, using the fact
  that all the branches have a common solution for $\beta = \beta_c$.
  


\end{proof}

\vskip.3cm
Let us describe in more details the case of
$a_i=a<1$ for all $i= 2, \dots, N$, with the help of some figures.
\begin{figure}[h]
  \includegraphics[scale=0.3]{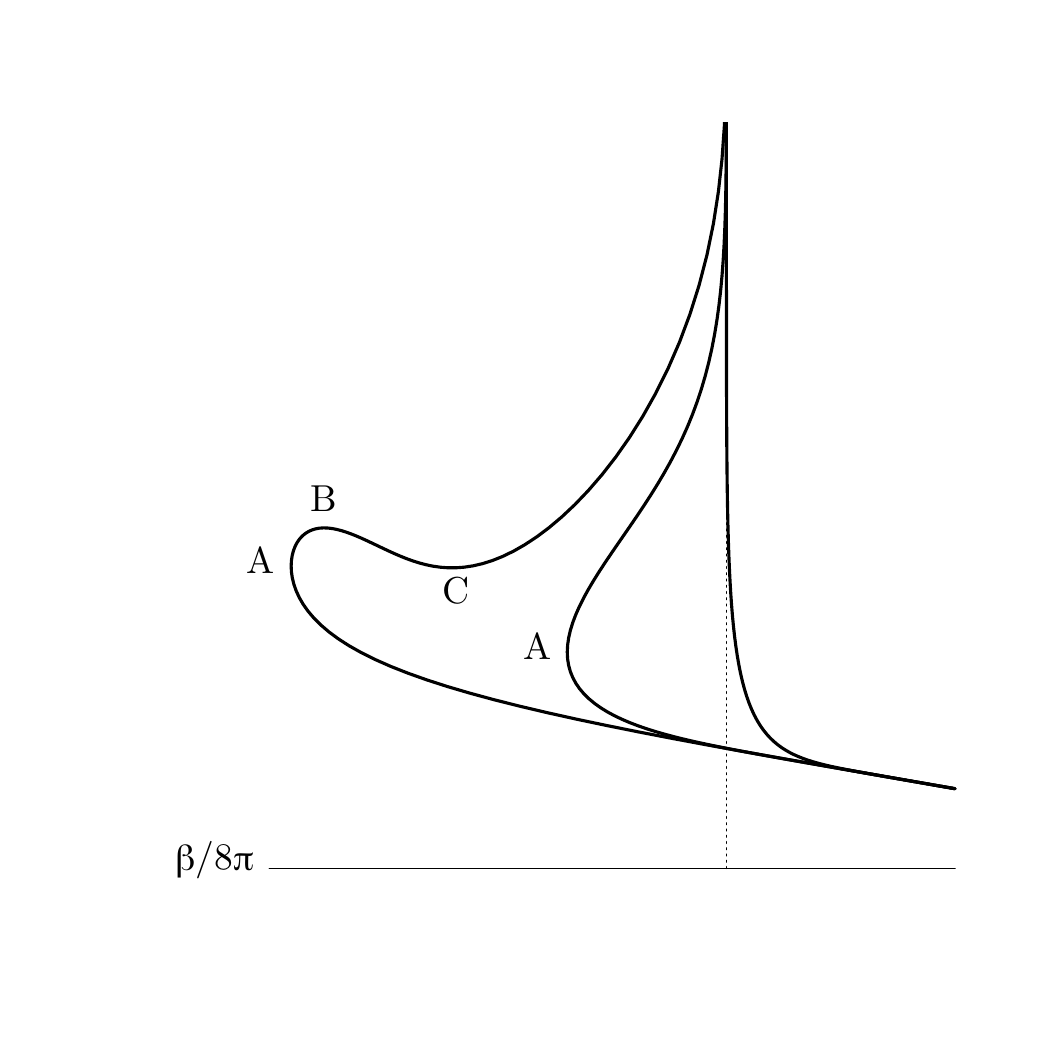}
  \caption{\small The qualitative behavior of the
    $0-$branch (MVP solution) in the plane $(\beta,E)$, for
    three increasing values of $a=a_i$, $i=2\dots N$ ($N \geq 3$ ). }
  \label{fig:b0}
\end{figure}
\begin{figure}[h]
  \includegraphics[scale=0.3]{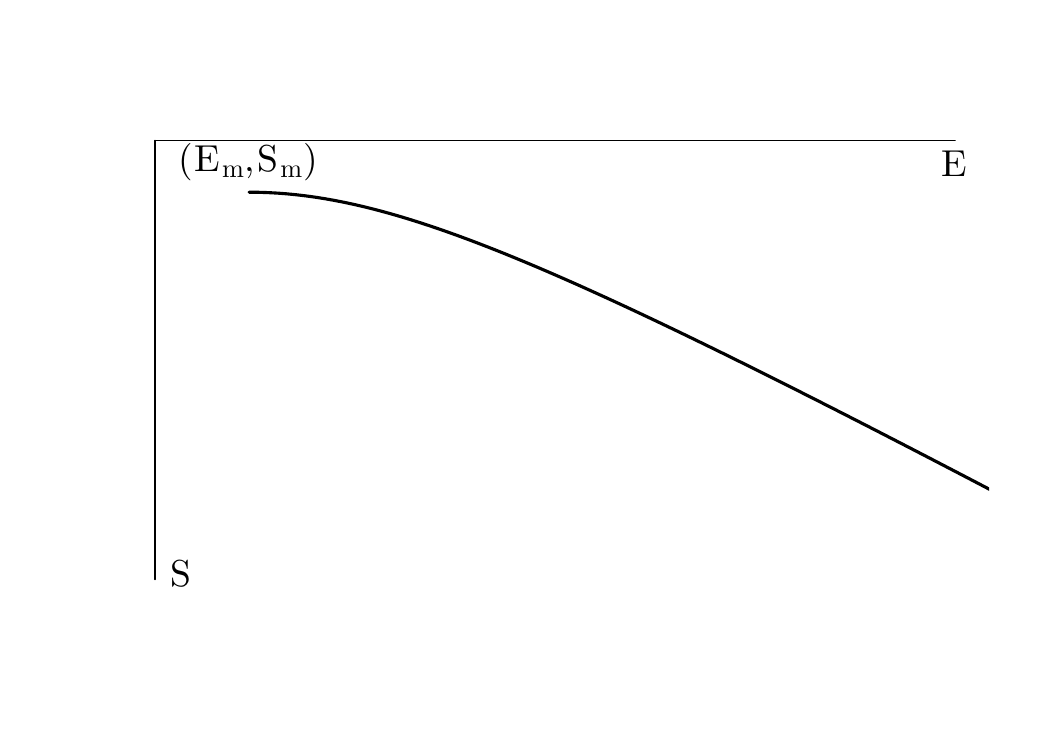}
  \includegraphics[scale=0.3]{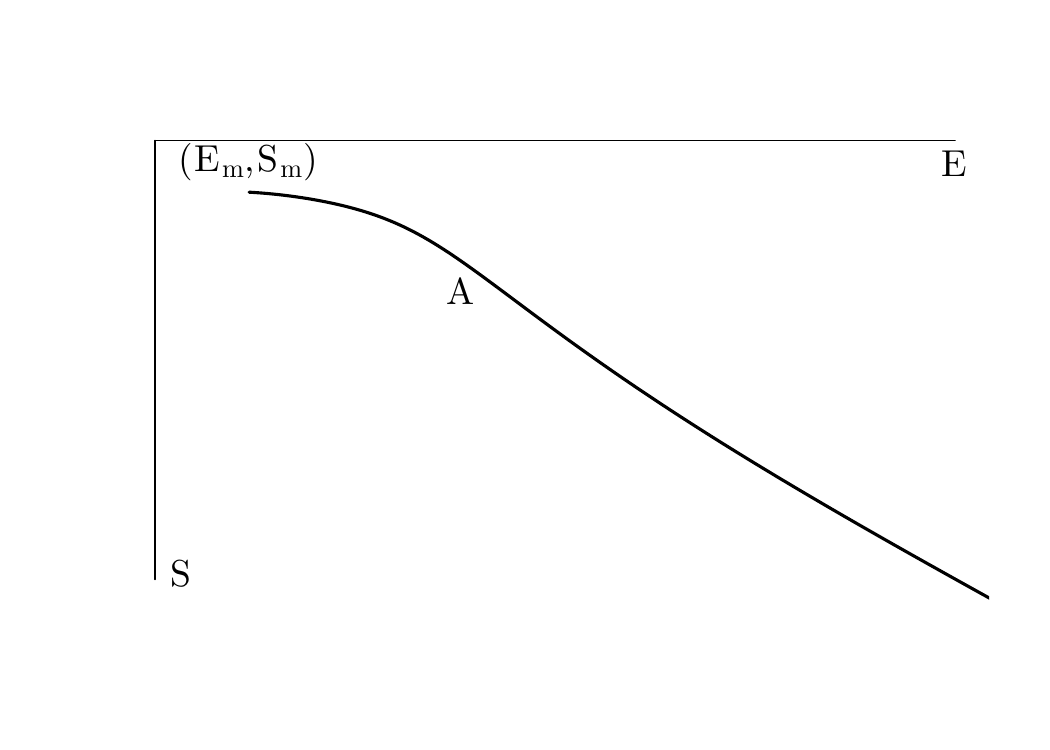}
  \includegraphics[scale=0.3]{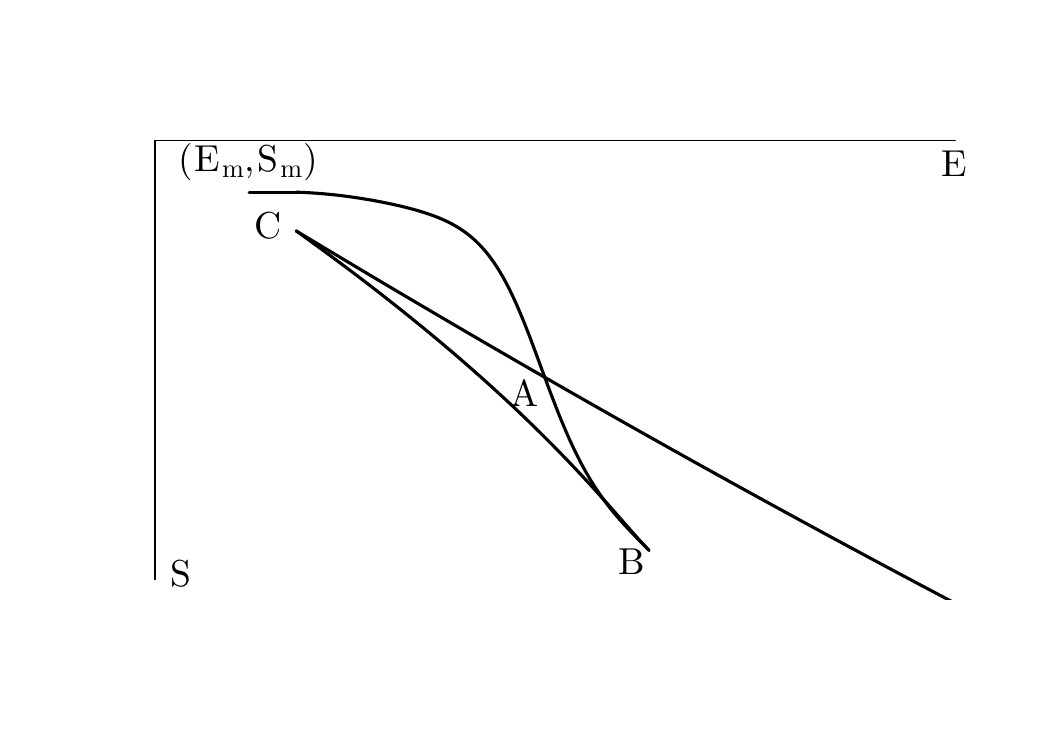}
  \caption{\small The qualitative behavior of the
    $0-$branch (MVP solutions) in the plane $(E,S)$, for
    the three increasing  values of  $a$ as in fig. \eqref{fig:b0}
  }
  \label{fig:s0}
\end{figure}
When $a$ increases from $0$ to $1$, the 
$0-$branch modifies  as in   fig. \ref{fig:b0},
indeed as $a \to 1^-$  it converges partially  to the $0-$merged-branch
and  partially  to the $1-$merged-branch, 
restricted to the region  $\beta \ge \beta_c$.
The entropy on the $0$-branch modifies as in fig.
\ref{fig:s0}.
In the point $A$
it becomes convex, 
between $B$ and $C$ is concave, and after $C$ is convex, and 
$(E(\mu),S(\mu))$ intersects itself.
Numerical evidence indicate that this happens in the first
concave piece, i.e. before the point $A$.
For given $E$, the solution $\ms S(E)$ of the MVP is the maximum value
of the entropy on the branches,
hence a phase transition appear.

We warn the reader that the numerical inspection is extremely delicate:
the distances between the three entropy branches in the last picture of
fig. \ref{fig:s0}
are very  small compared to the distances between the points $C$ and $B$.
For example, when $N=3$, $a=1-10^{-5}$, the ratio is about $2\times 10^{-3}$.




\subsection{High energy phase transitions}

In this subsection we show that for any $N\ge 2$
there exist disconnected sets for which
there are $N$ first order phase transitions of the entropy,
in the large energy region.
We consider small deformation of disks as follows.

Let $\eta>0$ be a small parameter, set $\eps = \eta^{1/2}$,
and consider the domain $\Lambda_{a,\eta}$
obtained by transforming the unit
circle with the conformal map 
$\C \ni z \to z + \eps z^3\in \C $, and then scaling so that
the area becomes $a$.

In the following proposition we give the thermodynamic quantities defined
in \eqref{def:ezmu} for the deformed domain $\Lambda_{a,\eta}$
(see the proof in Appendix \ref{appendice:A}).
\begin{proposition}
  For small $\eta$, $\Lambda_{a,\eta}$ is a fist kind domain,
  with
  \begin{equation}
    \label{eq:ezeta}
    \begin{aligned}
      &\mf e_{a,\eta}(\mu) =  e_{\eta}(\mu) = \frac 1{8\pi}
      \left( \frac 1{\mu^2} (-\mu - \log (1-\mu)) - \eta \tau(\mu)
      +o(\eta) \right),\\
      &\text{where } \tau(\mu) = \frac 2{(1-2\mu/3)^2},\\      
      &\mf z_{a,\eta}(\mu) = \frac {a}{1-\mu}(1+\eta \zeta (\mu) + o(\eta)),\\
      &\text{where } \zeta(\mu) =6\frac {1-2\mu+2\mu^2/3}{(1-2\mu/3)^2}.\\ 
    \end{aligned}
  \end{equation}
\end{proposition}

We consider $\Lambda = \bigcup_{i=1}^N \Lambda_{a_i,\eta_i}$,
with area $a_i$ near 1 and decreasing
with respect to $i$, and parameters $\eta_i$ small.
We proceed with the construction of the branch of the MFE solutions
in $\Lambda$ as in Proposition   \ref{propo:costruzione}, 
by choosing $\gamma$ small.
For any $i$ we have to solve
$$\mu_i (1-\mu_i) = \gamma a_i (1+\eta_i \zeta(\mu_i)+o(\eta_i))$$
which has two solutions
$$\mu_i^- = \gamma a_i (1+\eta_i) + o(\gamma)+o(\eta_i),
\ \ \mu_i^+ = 1- \gamma a_i (1-3\eta_i) +
o(\gamma)+o(\eta_i).$$
We consider the branch $\Cal B_i$  of solutions
by choosing $\mu_i = \mu_i^+$ and 
$\mu_j = \mu_j^-$ for all $j\neq i$. 
Since $\beta = -8\pi \sum_h \mu_h$ we have
$$\Cal S(\rho)
= -\log \gamma +\log \big(\sum_h \mu_h\big) - 16\pi \big(\sum_h \mu_h\big)
\Cal E(\rho).$$
In order to compare the entropy on the different branches we need
to express $\Cal S(\rho)$ in terms of  $E=\Cal E(\rho)$. 
We observe that
$\mu_j = \mu_j^- = o(\gamma)$ for $j\neq i$, and that
the dependence
of  $\mu_i (1-\mu_i)$
in $\gamma$ is exactly linear, hence
$$\log(1-\mu_i^+) = -\log \mu_i +
\log \gamma + \log a_i + \log (1+\eta_i z_{a_i,\eta_i}
(\mu_i^+)) +o(\eta_i).$$
We have
$$
\begin{aligned}
  &8\pi \big( \sum_h \mu_h \big)^2 E = 8\pi 
  \mu_i^2 e_{\eta_i}(\mu_i) + o(\gamma)  \\
  =& -1+(1-\mu_i^+) +
  \log (1-(1-\mu_i^+))
  -\log \gamma \\
  &- \log a_i - \log (1+\eta_i \zeta(\mu_i^+))
  -\eta_i \tau(\mu_i^+) 
  + o(\gamma) + o(\eta_i)\\
  =& -1 -\log\gamma
  - \log a_i - \eta_i (\zeta(\mu_i^+) + \tau
  (\mu_i^+)) + o(\gamma)+o(\eta_i)\\
  =& -1 -\log\gamma
  - \log a_i - 36\gamma a_i \eta_i + o(\gamma)+o(\eta_i)\\
\end{aligned}
$$
Then the entropy in terms of the energy $E$ for the  branch $\Cal B_i$ is
given by
$$
S_i  = - 8\pi (1 - (1-\sum_h \mu_h)^2) E +1+\log a_i
+\gamma \sum_j a_j (1+6\eta_j) - 2\gamma a_i (1- 24\eta_i) +o(\gamma)+
\sum_j o(\eta_{j}).
$$
  Since $-\log \gamma$ is of the order of $E$, we can drop
  the term  $(1-\sum_h \mu_h)^2 E$ which is $o(\gamma)$.
  To compare $S_i$, $i=1\dots N$, we set $a_i = 1 + \alpha_i \eta$, 
  $\eta_i = q_i \eta$, then
  $$
  S_i + 8\pi E  - 1 -\gamma \sum_j a_j (1+\eta_j) +2\gamma 
  = \eta ( \alpha_i + 2\gamma (24q_i-\alpha_i)) + o(\gamma) + o(\eta).
  $$
  By choosing suitable sequences of $\alpha_i$ and $q_i$
  we obtain that,  
  as $E$ increases (i.e. $\gamma$ decreases),
  the maximum of the entropies $S_i$ is reached, in the order,
  for $i=N, N-1,N-2,  \dots 1$.
  When $\ms S(E) = S_i$ we have 
  $$\pa_E \ms S(E) = \beta = -8\pi \sum \mu_h = 1
  -\gamma +\gamma \sum_j a_j (1+\eta_j) + \gamma \eta (24q_i - \alpha_i)
  +o(\gamma) + o(\eta),$$  
  then $\pa_E \ms S(E)$ 
  is discontinuous when $\ms S(E)$ goes 
  from $S_{i}$ to $S_{i+1}$.

  \begin{theorem}[High energy phase transitions]
    \label{teo:he}
    For any $N$, there exists a disconnected domain $\Lambda$
    for which $\ms S(E)$ has $N$ first order phase transitions,
    for large values of $E$.
    
  \end{theorem}

\section{$N-$dumbbell  domains}

The main result of this section is the construction
of connected bounded sets $\Omega$ for which
the entropy $\ms S_{\Omega}(E)$ has first order phase transitions,
similar to that in
theorems \ref{teo:le} and \ref{teo:he}.

We obtain the result in a perturbative setting, for which we have
to fix the notation.
We recall that if $\rho$
is a solution of MVP in a open bounded connected set $\Omega$, then,
by Theorem \ref{propo:mvp},
the stream function $\Psi$,
the inverse temperature $\beta$ and the normalization $Z$
are uniquely defined.
Moreover, the function $U=-\beta \Psi$ 
solves on $\Omega$
\begin{equation}
\label{eq:lambda}
-\lap U = \lambda \e^U
\end{equation}
with homogeneous Dirichlet condition on $\pa \Omega$, 
where $\lambda = -\beta / Z$ 
is a positive parameter for $\beta < 0$.
Conversely, 
if $U$ solves \eqref{eq:lambda},
we can define $Z=\int_\Omega \e^U$, $\beta = - Z\lambda$,
and then the stream function $\Psi = -U/\beta$,
solves the MFE with inverse temperature $\beta$.
In the sequel we will take for granted the relation between
$\rho$, $\Psi$, $U$, $\beta$, $Z$ and $\lambda$.

Let $\Lambda = \bigcup_{i=1}^{N} \Lambda_{i}$, where $\Lambda_{i}$ are
smooth open connected bounded sets of first kind which do
not intersect each other and let $C_{\eps}$ be the union of
channels of width $\eps>0$ connecting all the disjoint sets
$\Lambda_{i}$.
We consider $\eps_n \searrow 0$ and
a sequence of
{\it  $N$-dumbbell domain}
$\Omega_n$, \emph{i.e.} smooth connected open sets such that 
$ \Lambda \subset \Omega_n \subset (\Lambda \cup
C_{\eps_n})$.
If  $\Omega_{n+1} \subset \Omega_{n}$ for any $n$,
we say
that $\Omega_{n}$ is a \emph{decreasing sequence} of $N$-dumbbell domain
converging to $\Lambda$ and we write $\Omega_{n}\searrow \Lambda$.
We also suppose that
the measure
on the curve $\pa \Lambda$ of 
$\Omega_n \cap \pa\Lambda$ vanishes as $n\to +\infty$.


\begin{theorem}
\label{thm:convergence}
Consider $\Omega_{n}\searrow \Lambda$.
For any $E \in (0, + \infty)$ and any sequence $E_n \to E$, we have 
$\ms S_{\Omega_{n}}(E_n)  \to \ms S_{\Lambda}(E)$.
Let $\rho_n$ be a corresponding solution of the MVP in $\Omega_{n}$.
Up to subsequences,
$\rho_n$ converges, weakly in the sense of measure, to $\rho$ which 
is a solution of MVP in $\Lambda$
of energy $E$.
Moreover $\Psi_n\to \Psi$
strongly in $H^1$, 
$\beta_n\to \beta$, and $Z_n \to Z$.
\end{theorem}

We use the following lemma (see the proof of Proposition~2.1 in \cite{CLMP2}).
\begin{lemma}
  \label{lemma:e}
  Let $\rho_n$ be a sequence of probability densities on an
  open bounded set $\Omega$,
  with $\Cal S_\Omega(\rho_n)$ bounded from below and
  converging   weakly in $L^1(\Omega)$
  to a probability density $\rho$. Then
  $\ms E_\Omega(\rho_n) \to \ms E_\Omega(\rho)$.
\end{lemma}

\begin{proof}[Proof of Theorem \ref{thm:convergence}] 
  By Theorem \ref{propo:mvp} (MVP) for any $n \in \N$,
  $\ms S_{\Omega_{n}}(E_n)$ is attained in some $\rho_n$, which solves
  the MFE \eqref{eq:mfe} with inverse temperature $\beta_n$. Since we
  have
  \begin{equation*}
    \ms S_{\Lambda}(E)\leftarrow \ms S_{\Lambda}(E_n)
    \le \ms S_{\Omega_n}(E_n) = \Cal{S}_{\Omega_{n}}(\rho_n) \leq
    \ln |\Omega_{n}| \to  \ln |\Lambda| 
 \end{equation*}
 we get that, 
 up to
 subsequences, $\rho_n$ converges weakly   weakly in $L^1$
 to a probability density $\rho$.
 Moreover, we have
 $\ms S_{\Lambda}(E) \leq \limsup_{n} \Cal{S}_{\Omega_{n}}(\rho_{n} )
 \leq \mathcal{S}_{\Lambda}(\rho)$.  We now prove that
 $\Cal E(\rho) = E$, so that
 $\ms S_\Lambda(E) = \Cal S_\Lambda(\rho)$.

Let $G_{n}(x,y) $ and $G(x,y)$ be  the Green functions on
$\Omega_{n}$ and $\Lambda$,
respectively, extended to zero on the complementary sets
Given $k \in \N$, for any $n > k$ and $(x,y) \in
\Omega_{n} \times \Omega_{n}$ we have that 
$G_{n}(x,y) \leq G_{k}(x,y)$,
as follows from the positivity of $G_k(x,y)$ and the maximum principle
applied to $G_{n}-G_{k}$.
Then
$$2E_n \le 
\iint_{\Omega_{k} \times \Omega_{k}} \rho_n(x) G_{k}(x,y) \rho_n(y).$$
By the lemma \ref{lemma:e}, for any $k$
$$2E \le 
\iint_{\Lambda \times \Lambda} \rho(x) G_{k}(x,y) \rho(y).
$$
For any $x \in \Lambda$,  the function $y\mapsto G_{k}(x,y) - G(x,y)$
is harmonic in $\Lambda$  and 
$G_{k}(x,y) - G(x,y) = G_{k}(x,y)$ for any $y \in \partial  \Lambda$.
Hence 
 \begin{equation*}
   G_{k}(x,y) - G(x,y) = - \int_{\partial \Lambda}
   G_{k}(x,z) \frac{\partial G}{ \partial \nu} (z,y) \,d\sigma(z)
   \qquad \forall ( x,y) \in \Lambda \times \Lambda
 \end{equation*}
 where $\nu$ is the outer  normal to $\partial \Lambda$.
 For any $(x,z) \in \Lambda\times \pa \Lambda$ the sequence 
 $G_k(x,z)$ is decreasing, and, by the construction of $\Omega_n$, 
 vanishes for $k\to +\infty$, a.e. with respect to $\de\sigma$.
 Therefore
 $G_{k}(x,y) - G(x,y)  \to 0 $  for any
 $(x,y) \in \Lambda \times \Lambda$. Finally,  by dominated convergence 
 $$  2E \le  \iint_{\Lambda\times \Lambda} \rho(x) G_k(x,y) \rho(y)\to 
 \iint_{\Lambda\times \Lambda} \rho(x) G(x,y) \rho(y)
 = 2\Cal E_\Lambda(\rho).$$
 To prove that $\Cal E_\Lambda(\rho) \le E$, we set 
 $\Psi_n = (- \Delta)^{-1} \rho_n $  and $\Psi = ( - \Delta)^{-1}
 \rho \in H^1_{0}(\Lambda)$.
 We have that $\Psi_n \rightharpoonup   \Psi$
 weakly  in $ H^1$, hence 
 $ \Cal E_{\Lambda} (\rho) =\frac{1}{2} \int_{\Lambda} | \nabla  \Psi |^2 \leq \liminf_n  \Cal E_{\Omega_{n}} (\rho_n)  = E$.

 Since $E=\Cal E_\Lambda(\rho)$ and $\ms S_\Lambda(E) \le  \Cal
 S_\Lambda(\rho)$,  we have that $\Cal S_{\Lambda}(\rho) = \ms
 S_{\Lambda}(E)$, and consequently $\ms S_{\Omega_n}(E_n) \to
 \ms S_\Lambda(E)$.
 The convergence of the energy, i.e. of $\|\Psi_n\|_{H^1_0}$,
 assure that  $\Psi_n \to \Psi$ strongly in $H^1$.

 We now prove the convergence of $\beta_n$ to $\beta$,
 which is the inverse temperature for 
 the MFE \eqref{eq:mfe} for $\Psi$.
 For any subsequence, 
 there exists a subsequence such that $\Psi_k(x) \to \Psi(x)$
 a.e. in $\Lambda $. Since $\Psi$ is continuous, positive on $\Lambda$
 and $0$ in $\pa \Lambda$, 
 the closed set $C= \Psi^{-1} ([0, E/2] )$ has positive measure.
 For any $n \in \N$, set $C_{n} =  \{x
 \in \bar \Lambda:\,
 0\leq \Psi_k(x)  \leq E_k,   \, \forall k \geq n \} \cap C$. We have 
 $C_n \subseteq C_{n+1} $ and   $\bigcup_{n \in \N} C_{n} = C$
 up to a set of zero measure, therefore $|C_n| \to |C| $ as $n \to + \infty$.
 Since $\int_{C_n} \rho_k \to \int_{C_n}\rho$,
 for $n$ sufficiently large,
 for any $k\ge n$, we have 
 $0 < c \le \int_{C_n} \rho_k \le 1$.
 Now we prove that $Z_k$ lies in a compact subset of $(0,+\infty)$.
 Note 
 that if $\beta_k \le 0$, $Z_k \ge |\Omega_k|$,
 while if $\beta_k > 0$, $Z_k \in (0,|\Omega_k|)$,
 and recall that 
 \begin{equation*}
   \mathcal{S}_{\Omega_{k}}(\rho_k) = 2  \beta_k E_k+ \ln Z_k.
 \end{equation*}
 If $\beta_k \le 0$ we have 
 \begin{align*}
   c \le \int_{C_n} \rho_{k}
   \leq \frac{ \text{e}^{- \beta_k E_k }}{Z_k} |C_n |  =
   \frac{ \text{e}^{-\frac{1}{2}  \mathcal{S}_{\Omega_k}
   (\rho_k)  }} {Z_k^{\frac{1}{2}}} |C_n |
 \end{align*}
 then $Z_k$ is bounded from above. If $\beta_k> 0$
 \begin{align*}
   1 \geq  \int_{C_n} \rho_{k} \geq \frac{ \text{e}^{- \beta_k E }}{Z_k}
   |C_n |  =  \frac{ \text{e}^{-\frac{1}{2}  \mathcal{S}_{\Omega_k}(\rho_k)  }}
   {Z_k^{\frac{1}{2}}} |C_n |
 \end{align*}
 then $Z_k$ is bounded from below by a positive constant.
 Consequently, up to subsequence, $Z_k \to \tilde Z\in (0,+\infty)$, and then,
 from the relation with the entropy,
 $\beta_k \to \tilde{\beta}\in \R$.
 Moreover 
 $\rho_k(x) \to
 \frac{ \text{e}^{- \tilde{\beta} \psi (x)}}{\tilde{Z}}  =
 \rho(x) =  \frac{ \text{e}^{- \beta \psi (x)}}{Z}  $ a.e. in
 $\Lambda$  which implies $\tilde{\beta} = \beta $ and 
 $ \tilde{Z} = Z$.
\end{proof}


\vskip.3cm

Now we consider $\Lambda$ the union of $N$ disjoint disks,
as in section \ref{sez:piudischi}.
We prove that for $N$ large 
the $N$-dumbbell
domains $\Omega_{n}$ have branches of
solutions of the MFE close to those obtained for the disconnected
set $\Lambda$. The proof can be easily adapted to the case of
deformed disks considered in Theorem \ref{teo:he}.

We first consider the solutions of \eqref{eq:lambda}.
Note that
the inverse temperature  $\beta$ is a function of
$\lambda$ and of the solution $U$, 
hence different solutions of \eqref{eq:lambda} with the same $\lambda$ provide
solutions of the MFE with different $\beta$.
Also the energy does not depend only on the parameter $\lambda$,
and its expression  in terms  $U$ if given by
\begin{equation}
  \label{eq:energiaU}
  E = \frac 1{2\lambda^2 Z^2}  \int |\grad U|^2 = 
  \frac 1{2\lambda Z^2}  \int U\e^{U}.
\end{equation}
To clarify
what can happen, 
recall that we have constructed the branches of solutions 
in the parameter $\mu\in [0,1]$. Using the relation 
between $Z$ and $\mu$ in \eqref{eq:termo}
(and in \eqref{eq:termo2} for the case of merged branches),
we have that
$$\lambda = \frac {8\pi}{a_1} \mu (1-\mu),$$
hence, in particular, $\lambda$ 
is proportional to $\gamma$ 
in \eqref{eq:mui}.
The parameter $\lambda$ has the maximum value 
$\lambda_c=2\pi/a_1$.  The values of $\mu\in (0,1/2)$
describes the right part of the branches in the plane $(\beta,E)$,
while $\mu\in (1/2,1)$ describe the left part.
So that, fixed $\lambda \in (0,\lambda_c)$,
there exist two solutions of
eq. \eqref{eq:lambda}
for any branch.
We show in figure \ref{fig:lambda} how the branches of solutions
appear in the $(\lambda,E)$ plane.
We consider $N=3$, and different values of $a_2$ and $a_3$, 
(see figure \ref{fig:tre12} and \ref{fig:locale}
for the analogous figures
in the $(\beta,E)$ plane).

\begin{figure}[h]
  \centering
  \parbox{15cm}{
    \includegraphics[scale=0.27]{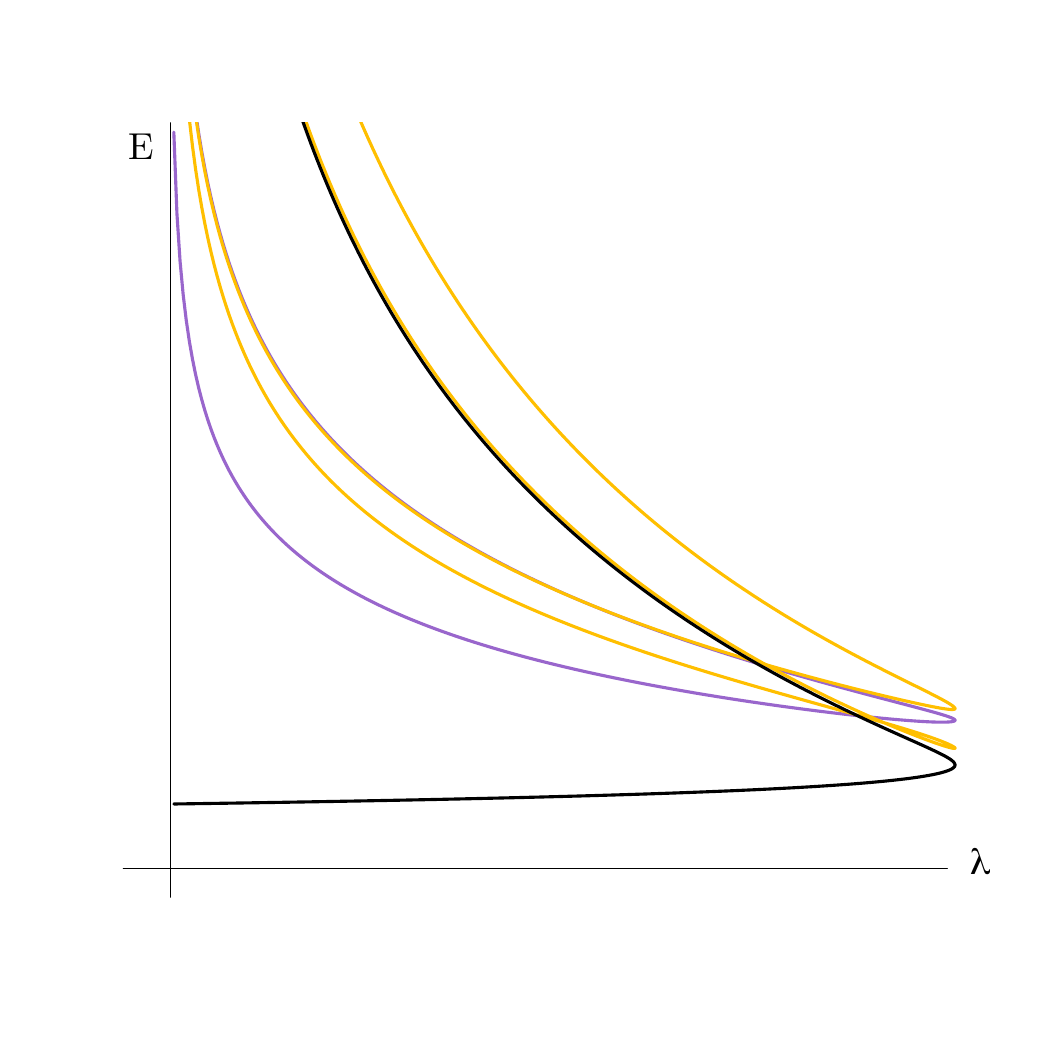}
    \includegraphics[scale=0.27]{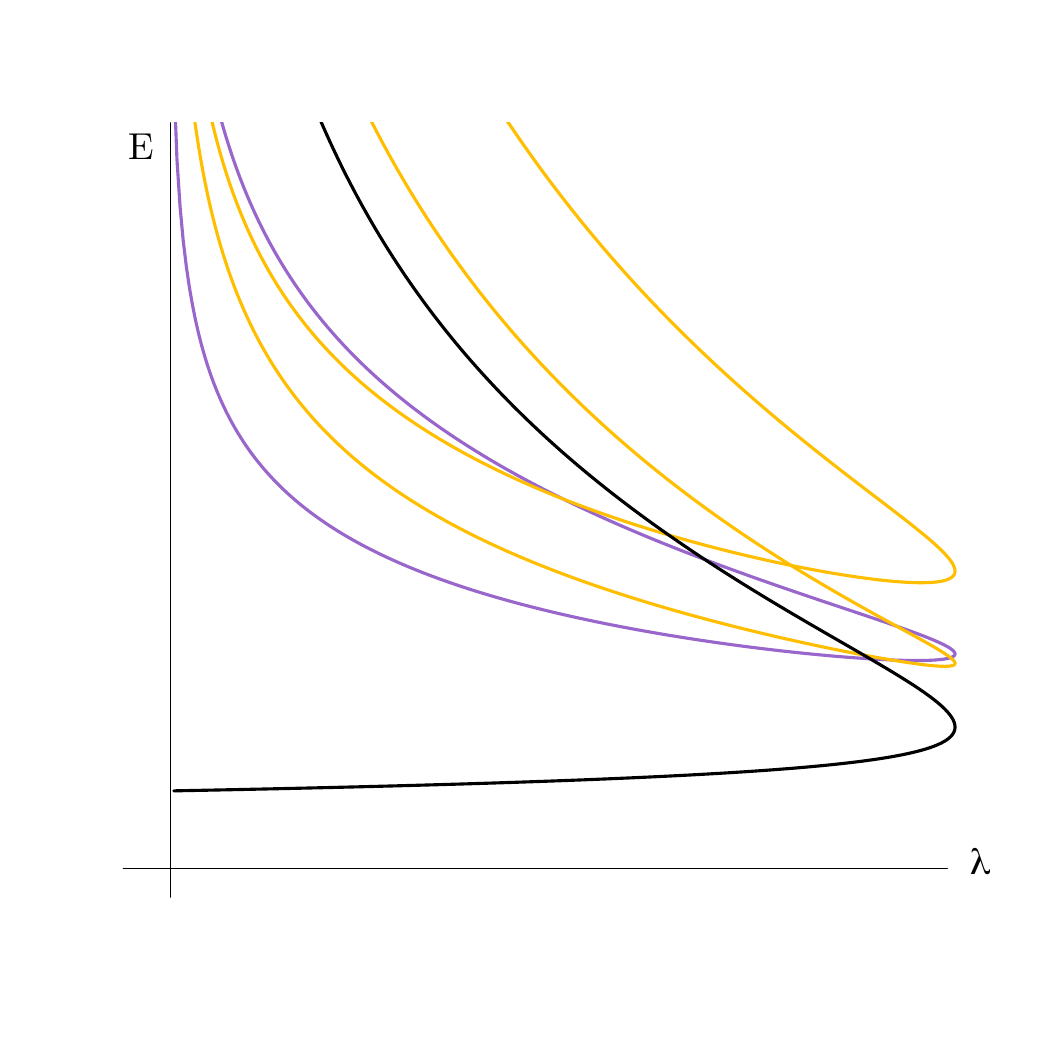}
    \includegraphics[scale=0.27]{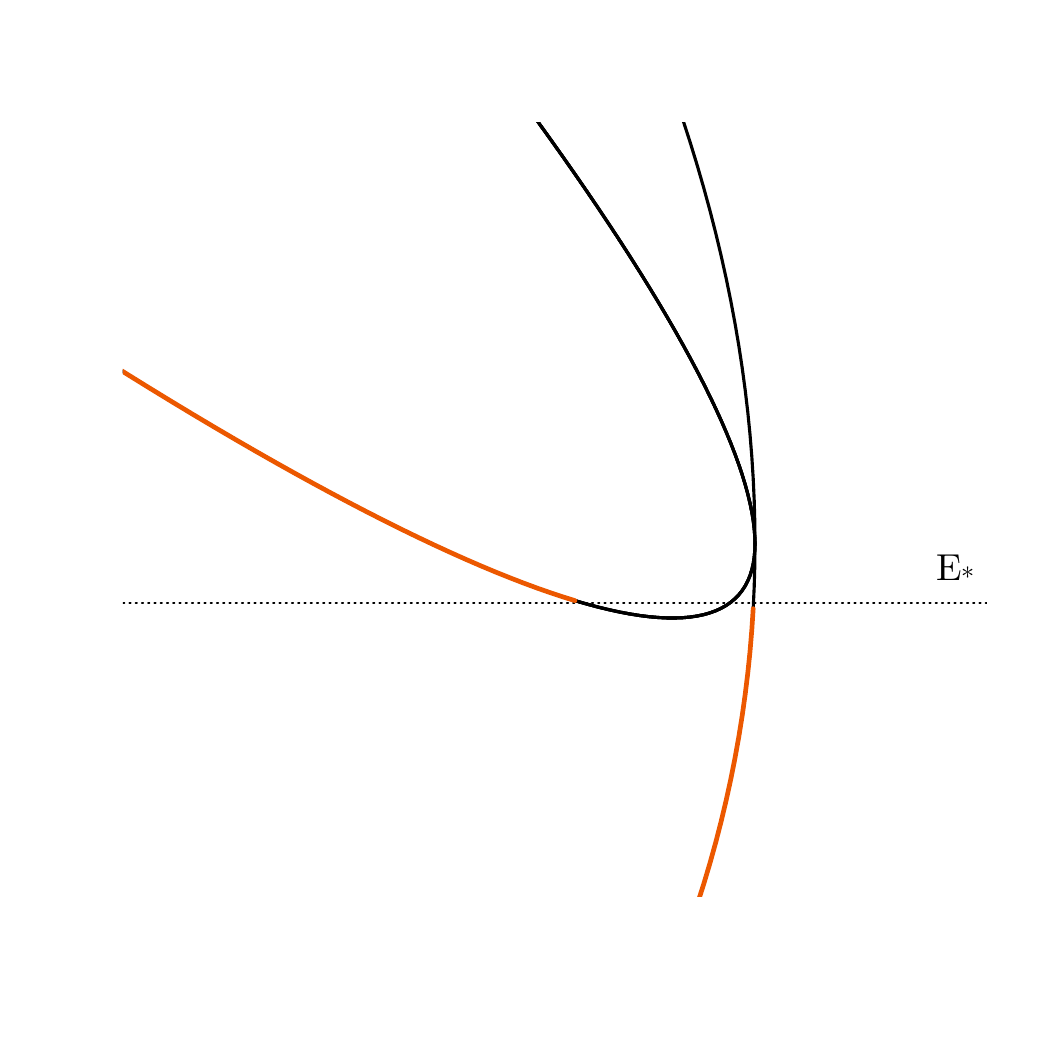}
  }
  \caption{\small The branches of solutions in the $(\lambda,E)$ planes.
    On the left $(a_2,a_3)=(0.6,0.2)$ and in the center $(a_2,a_3) =
    (0.96,0.5)$; in this two graphs,
    the lowest branch is the $0$-branch,
    the next one is the $1-$branch.
    On the right
    we show the branches 
    near $\lambda_c$
    in the case 
    $(a_2,a_3) = (1,1)$;
    the solution of
    the MVP jumps at
    $E=E_*$ from the $0-$merged-branch to the $1-$merged-branch.
  }
  \label{fig:lambda}
\end{figure}

It is clear that it is convenient to parametrize
the solutions of
the MFE \eqref{eq:mfe} on $\Lambda$ on a
$k-$ branch (or on a $k-$merged-branch)
as $\{  \Psi_{\mu} \}_{\mu \in (0,1)}$.
We set 
$\lambda(\mu) = \frac{ - \beta(\mu) }{Z(\mu) } =  
\frac
{8\pi }{a_1} (1-\mu)\mu =  \frac{ - M_i \beta }{Z_i}$,
and 
$U_{\mu} = - \beta(\mu) \Psi_{\mu}$.
We have the following result.

\begin{proposition}
  \label{prop:gluing}
  Let  $\Omega_{n}\searrow \Lambda$, as $n \to + \infty$.
  For any compact set $W \subset (0, 1) \setminus \{\frac{1}{2}\}$
  there exists  $r_W >0$ such that 
  for any $r\in(0,r_W)$ there exists $n_r$ such that
  for any   $n \geq n_r$   and 
  $\mu \in  W $ there exists  
  $U^{n}_{\mu} \in 
  H^1_0 (\Omega_n)  $ which  is the unique solution  in
  $B_{r}(U_{\mu})$ of \eqref{eq:lambda} in $\Omega_n$
  with $\lambda = \lambda(\mu)$.
\end{proposition}

\begin{proof}
We set   $U^{(i)}_{\mu }= U_{\mu}|_{\Lambda_i } $, for  $i =1, \dots, N$.
We have 
$ - \Delta U^{(i)}_{\mu} = \lambda(\mu) \text{e}^{U^{(i)}_{\mu} }  $
in $ \Lambda_i $.
Consider  the linear operator 
$T^{(i)}_{\mu}:  H^1_0(\Lambda_i )  \to H^{-1} (\Lambda_i ) $ 
given  by 
\begin{equation*}
	\langle T^{(i)}_{\mu} f, g \rangle = \int_{\Lambda_i} \nabla f \cdot \nabla g- \lambda(\mu) \int_{\Lambda_i} \text{e}^{ U^{(i)}_{\mu}} f  g  , \qquad  \qquad \forall f, g \in H^1_0(\Lambda_i ).
\end{equation*}
By \cite{S} we know that
if $\mu \in W$, $\text{Ker}\, T^{(i)}_{\mu} = \{0\} $,
and if $\mu_n \in W$ then $U^{(i)}_{\mu_n} \to U^{(i)}_{\mu_0}$
in  $C^{2+\alpha}(\Lambda_i) \cap C(\bar{\Lambda}_i)$
(up to subsequence) for some $\mu_0 \in W$.  

We extend $U_\mu$ to zero in the complementary of $\Lambda$,
and we define  $T^n_{\mu}  : H^1_0(\Omega_{n} ) \to H^{-1} (\Omega_{n} )$
as follows
\begin{align*}
  \langle T_{\mu}^n  f, g \rangle
  &= \int_{\Omega_{n}}
    \nabla f \cdot \nabla g -\lambda(\mu)
    \int_{\Omega_{n}} 	\text{e}^{U_{\mu}} f g \\
  &= \int_{\Omega_{n}}
    \nabla f \cdot \nabla g - \lambda(\mu)  \sum_{i=1}^N \int_{\Lambda_i} 
    \text{e}^{U^{(i)}_{\mu}} f g -
    \lambda(\mu)  \int_{\Omega_{n} \setminus \Lambda } fg. 
\end{align*}
We now prove that 
there exist $n_W>0$ and  $C_W>0$ such that
for  any   $n\ge n_W$ and any 
 $  \mu \in  W$ 
 \begin{equation*}
   \| T_{\mu}^{n}  f \|_{H^{-1}} \ge  C_W
   \| f \|_ {H^1_0} \qquad \forall f \in  H^{1}_0(\Omega_{n}) .
\end{equation*}	
If not,  there exists a diverging
sequence $n_k\in \N$, a sequence $\mu_k  \to \mu_0  \in  W$,   
and $h_k \in   H^1_0(\Omega_{n_k}) $,  
$\| h_k \|_{H^1_0}=1$ such that  
$\| T^{n_k}_{\mu_k}  \, h_k \|_{H^{-1}}\to 0 $. 
Then, up to subsequence, 
$h_k\rightharpoonup h_0$ 
weakly in $H^1$ and  $h_k  \to h_0$ strongly in $L^2$ 
Moreover,
since $\text{Int} \, (\bigcap_{n=1}^{+\infty} \Omega_{n_k})= \Lambda$,
we have that
$h_0 \in H^1_0(\Lambda)$ and    
\begin{align*}
  \langle T_{\mu_k}^{n_k}  h_k, g \rangle 
  &= \int_{\Omega_{n_k}} \nabla h_k \cdot \nabla g -
    \lambda(\mu_k)
    \sum_{i=1}^N \int_{\Lambda_i} 
    \text{e}^{U^{(i)}_{\mu_k}} h_k g -
    \lambda(\mu_k)  \int_{\Omega_{n_k} \setminus \Lambda } h_k g 
    \to \sum_{i=1}^N \langle T^{(i)}_{\mu_0}  h_0, g \rangle 
\end{align*}
for any $g \in H^1_0(\Lambda) $. Hence, setting
$h^{(i)}_0 = h_0|_{\Lambda_i} \in H^1_0 (\Lambda_i)$ we have
$h^{(i)}_0 \in \text{Ker}\, T^{(i)}_{\mu_0} = \{0\}$. 
We get a contradiction,  since 
\begin{align*}
  1
  &=  \| h_k \|^2_{H^1_0} = \langle T_{\mu_k}^{n_k}  h_k , h_k \rangle +
    \lambda(\mu_k)  \int_{\Omega_{n_k}}\text{e}^{U_{\mu_k} }
    |h_k|^2 \leq \| T_{\mu_k}^{n_k}  \,  h_k \|_{H^{-1}} +
    C \| h_k \|^2_{L^2}  \to 0. 
\end{align*}
As a consequence,  for all $n> n_W$, we have $\text{Ker} \ T_\mu^n =\{0\}$
and $T_\mu^n$ is invertible, 
since $T_\mu^n \lap^{-1} $ is a Fredholm operator in $H^{-1}$.
In particular, the inverse $S^n_\mu$ is uniformly bounded by $C_W^{-1}$
for $\mu \in W$.

Now, 
for any $n>n_W$ we consider 
the $C^1$-map $F^n_{\mu} : H^1_0 (\Omega_{n}) \to H^{-1} (\Omega_{n}) $ 
given by 
\begin{equation}
  \label{eq:fmun}
  \langle F^n_{\mu} (v), h \rangle =
  \int_ {\Omega_{n}} \nabla ( U_{\mu} + v) \cdot
  \nabla h - \lambda(\mu) \int_ {\Omega_{n}}  \text{e}^{U_{\mu} + v}h
    \qquad \forall v,h \in H^1_0 (\Omega_{n}).
\end{equation} 
Clearly  
$\de F^n_{\mu} (0) = T^n_{\mu}$  and 
\begin{align}
    \label{eq:fmun0}
  \langle F^n_{\mu} (0), h \rangle
  & = \int_ {\Omega_{n}} \nabla U_{\mu} \cdot
    \nabla h - \lambda(\mu) \int_ {\Omega_{n}}  \text{e}^{U_{\mu}}h
    =  \int_ {\pa (\Omega_{n} \setminus \Lambda) \cap \partial
    \Lambda} \frac{\partial}{\partial \nu} U_{\mu} \, h - \lambda(\mu)
    \int_ {\Omega_{n} \setminus \Lambda }h.
\end{align}
Since $U_{\mu} \in H^2(\Lambda) $ we have in particular
$\frac{\partial}{\partial \nu} U_{\mu}{|_{\partial \Lambda }} \in
L^2(\partial \Lambda) $ and since
$|\partial(\Omega_{n} \setminus \Lambda) \cap \partial \Lambda |
\to 0$ we get $\|F^n_{\mu}(0) \|_{H^{-1}}\to 0$ as $n\to +\infty$,
uniformly in $\mu\in W$.

We are interested in the zeros of $F_{\mu}^n $,  or equivalently, 
the fixed points of  the smooth map $G_{\mu}^n : H^1_0 (\Omega_n)
\to H^1_0 (\Omega_n)  $ 
with 
\begin{equation*}
G_{\mu}^n (v)= v -  S_{\mu}^n F_{\mu}^n(v).
\end{equation*} 
We have  $\| G_{\mu}^n (0) \|_{H^1_0}  \leq C_W^{-1}
\| F_{\mu}^n (0) \|_{H^{-1}} \to 0$,
as $n \to \infty$, and  $\de G_{\mu}^n(0) = 0;$
moreover,  by  the Moser-Trudinger inequality,   for any $v \in  H^1_0
(\Omega_{n}) $
\begin{align}
  \label{eq:dG}
	\|\de G_{\mu}^n(v) \| &= \|\de G_{\mu}^n(v) -  \de G_{\mu}^n(0)\| = 
\| S_{\mu}^n (\de F_{\mu}^n (v) - \de F_{\mu}^n (0))\| \\
& \leq C_W^{-1} \|\de F_{\mu}^n (v) - \de F_{\mu}^n (0)\|
 \leq C \| \text{e}^{U_{\mu} }( \text{e}^{v} -1) \|_{L^2} \\
 &\leq C \| v \text{e}^{ |v|}  \|_{L^2} \leq 
C \| v \|_{L^4} \| \text{e}^{|v|} \|_{L^4} \leq  C \|v\|_{H^1_0} 
\text{e}^{\frac{1}{4 \pi}\| v \|^2_{H^1_0}}.
\end{align} 	
for some constant $C>0$ (independent on $n$ and $\mu \in W$). 
Hence there exists $r_W >0$ such that $\|\de G_{\mu}^n(v)\| < 1$
in $B_{r} = \{ \|v\|_{H^1_0} \leq r \}$ for any $r\le r_W$.
Since $\|G_\mu^n(0)\|_{H^1}\to 0$ uniformly in $W$,  
for any $r\in (0,r_k)$ there exists $n_r$ independent
on $\mu\in W$, such that for any $n\ge n_r$ the map $G_\mu^n(v)$ is a strict
contraction in $B_r$.
Therefore there exists (unique)
$v_\mu^n \in B_{r} $, fixed point for $G_{\mu}^n$.
Namely $U^n_\mu = U_\mu + v_\mu^n$ solves 
\eqref{eq:lambda}, uniquely
in $B_{r}(U_\mu)$, and $v_\mu^n\to 0$ in $H^1_0$
as $n\to +\infty$, uniformly for $\mu \in W$.
\end{proof}

The following lemma allow us to re-parametrize in terms of the energy 
the branch of solutions 
$\{ U_\mu^n\}_{\mu \in W}$ on $\Omega_n$.

\begin{lemma}
  \label{lemma:inv}
  Let $W$ be a closed interval in $(0,1)\setminus\{1/2\}$,
  and $\{U_\mu\}_{\mu\in W}$ a branch of solutions for the MFE on $\Lambda$.
  Let 
  $E(\mu)$ be the energy of $U_\mu$. 
  Assume that 
  $E'(\mu)\neq 0$ for $\mu\in W$, and denote by $\mu(E)$ its inverse.
  Let $\{U^n_\mu\}_{\mu \in W}$ be the solutions of
  \eqref{eq:lambda} in $\Omega_n$
  given by Proposition \ref{prop:gluing}.
  Then, for any $n$ sufficiently large
  there exists a $C^1$ function $\mu^n(E)$ such
  that the energy in \eqref{eq:energiaU} associated to
  $U^n_{\mu^n(E)}$ is $E$.
  Moreover $\mu^n(E)$ converges uniformly with its derivative 
  to $\mu(E)$, the 
  entropy associated to $U^n_{\mu^n(E)}$ 
  is a $C^2$ function on $E$, and converges uniformly 
  with its derivatives to the entropy associated to $U_{\mu(E)}$. 
\end{lemma}

\begin{proof}
  We consider the $C^1$ map $(\mu,v) \mapsto F_\mu^n(v)$
  defined in \eqref{eq:fmun}.
  By Proposition \ref{prop:gluing}, $\mu \mapsto v_\mu^n = U^n_\mu - U_\mu$
  is the curve of solutions of $F_\mu^n(v)=0$ we can obtain via the
  implicit function theorem, then it is a $C^1$ function
  from $W$ to $H^1$.

  We now show that $\pa_\mu v_\mu^n\to 0$ in $H^1_0$, uniformly in $W$.
  By the implicit function theorem
  $\pa_\mu v_\mu^n = -\de F_\mu^n(v_\mu^n)^{-1}[(\pa_\mu F_\mu^n)(v_\mu^n) 
  ].
  $
  Since the operator $\de F_\mu^n(v_\mu^n)^{-1}:H^{-1}(\Omega_n)\to
  H^1_0(\Omega_n)$
  is uniformly bounded, we have to prove that
  $\|(\pa_\mu F_\mu^n)(v_\mu^n)\|_{H^{-1}}$ vanishes uniformly in $W$.
  For $h\in H^1_0(\Omega_n)$ we have
  \begin{equation*}
    \begin{aligned}
    &  \langle (\pa_\mu F^n_{\mu}) (v_\mu^n), h \rangle =
  \int_ {\Lambda} \nabla \pa_\mu U_{\mu} \cdot
  \nabla h - \lambda'(\mu) \int_ {\Omega_{n}}  \text{e}^{U_{\mu} + v_\mu^n}h
  -\lambda(\mu) \int_ {\Omega_{n}}  \text{e}^{U_{\mu} + v_\mu^n}\pa_\mu U_\mu h\\
  &=
  \int_{\Omega_n} \pa_\mu (\lambda(\mu)\e^{U_\mu})(1-\e^{v_\mu^n}) h
 - \lambda'(\mu) \int_{\Omega_n \setminus \Lambda} h +
\int_{\pa(\Omega_n \setminus \Lambda) \cap \pa\Lambda}
  \pa_\nu (\pa_\mu U_\mu) h.
\end{aligned}
\end{equation*}
Proceeding as  in the proof of Proposition \ref{prop:gluing}
(see the $L^2$ estimate of $\e^v -1$ in  \eqref{eq:dG},
and the proof that $\|F^n_\mu (0)\|_{H^{-1}}\to 0$
after \eqref{eq:fmun0})
we obtain that 
$\|(\pa_\mu F_\mu^n)(v_\mu^n)\|_{H^{-1}}$ vanishes uniformly in $W$.

As a consequence, $Z^n(\mu) = \int_{\Omega_n} \e^{U_\mu^n}$
is a regular function in $\mu$,
and $\pa_\mu Z^n(\mu)$ converges uniformly in $W$ to
$\pa_\mu Z(\mu)$.
We also note that
$$\pa_\mu  \int_{\Omega_n} |\grad U_\mu^n|^2
= 2\pa_\mu Z^n(\mu)$$
We denote with $E^n(\mu)$  the energy of  $U^n(\mu)$.
Since $\lambda(\mu)$ and $Z^n(\mu)$ do not vanish in $W$,
from
\eqref{eq:energiaU}
we get that $\pa_\mu E^n(\mu)$ converges uniformly in $W$
to $\pa_\mu E(\mu)\neq 0 $, hence for $n$ sufficiently large
also $E^n(\mu)$
is invertible, with inverse in $C^1$.
Recalling that
$\pa_E \ms S = \beta$,
we conclude the proof by noticing that
$\beta^n(E) =  - \lambda(\mu^n(E)) Z^n(E)$
is a $C^1$ function.
\end{proof}

In the hypothesis of lemma \ref{lemma:inv}
we can define $U(E) = U_{\mu(E)}$ and,
for $n$ sufficiently large, also 
$U^n(E) \coloneqq U^n_{\mu^n(E)}$.
The domain of definition of $U(E)$ and $U^n(E)$ are not
the same, but,
from Proposition \ref{prop:gluing}, we get the following result.

\begin{corollary}
  \label{corollario-E}
  In the hypothesis of lemma \ref{lemma:inv},
  for any closed interval $I$
  contained in the interior of $K=\{E(\mu):\, \mu \in W\}$,
  there exists $r_I>0$
  such that for any $r\in (0,r_I)$ there exists $n_r$ such
  that for any $n\ge n_r$, and for any $E\in I$,
  $U^n(E)$ is the unique solution
  of \eqref{eq:lambda} with energy $E$ in $B_r(U(E))$.
\end{corollary}

\vskip.3cm

We now show that there are bounded connected sets
for which the entropy has a first order phase transition.
We use the perturbative construction in Proposition
\ref{prop:gluing}, 
applied to the case of disks of different area, in
the hypothesis of Theorem \ref{teo:le}, 
for which the transition occurs for some $E_*$.
The construction can be adapted also to the case of Theorem \ref{teo:he}.

\begin{theorem}
  \label{thm:transizione}
  Consider $\Lambda$
  union of $N\ge 3$ disjoint disks with area $1=a_1> a_2 \ge \dots a_N$,
  with $a_i$ sufficiently close to $1$ as in Thm. \ref{teo:le}.
  If $\Omega_{n}\searrow \Lambda$,
  for $n$ sufficiently large
  there exists $E_*^n$ such that for $E$ in a neighborhood of $E_*^n$, 
  the MPV  in $\Omega_n$ has a unique solution
  if $E\neq E_*^n$, 
  has two solutions for $E=E_*^n$
  and 
  $\ms S_{\Omega_n}(E)$
  has a first order phase transition in $E=E_*^n$.
\end{theorem}

\begin{proof}
  We fix a closed interval $I\subset (E_0,{\bar E}_c)$,
  where $E_0$ and ${\bar E}_c$ are defined in the proof of
  Theorem \ref{teo:le}
  (see the graph on the right in Fig. \ref{fig:locale}).
  We also assume that $I$ contains  
  $E_*$ in its interior.
  
  We indicate with
  $\{\rho^\sigma(E)\}_{E\in I}$, $\sigma=0,1$,
  the solutions of the MVP on the $0-$branch,
  with $E<E_*$ and $E>E_*$ respectively.
  and we indicate with $\Psi^\sigma(E)$,
  $\beta^\sigma(E)$, $Z^\sigma(E)$, $\lambda^\sigma(E)$,
  $U^\sigma(E)$ the related functions and parameters  
  in the MFE \eqref{eq:mfe} and \eqref{eq:lambda}.
  
  For any $E\in I$,
  let  $\tilde {\Cal M_n}(E)\subset P_{\Omega_,E}$ be 
  the set of the solutions of the MVP in $\Omega_n$ with energy $E$,
  end let $\Cal M_n(E)\subset H^1_0(\Omega_n)$ be the corresponding subset of 
  the functions $U$.
  We claim that 
  for any $r >0$ there exists $n_{I,r} \in \N$
  such that for any $n \geq n_{I,r}$ and any $E \in I$
  \begin{equation}
    \label{eq:intorni}
    \Cal M_n(E) \subset   B_{r}(U^0(E)) \cup B_{r}(U^1(E)).
  \end{equation}
  If not, there exist sequences
  $E_k\in I$ and $\rho_k\in P_{\Omega_k,E_k}$,
  whose corresponding $U_k$ is not in
  $B_{r}(U^0(E_k)) \cup B_{r}(U^1(E_k))$.
  There exists a subsequence such that $E_k\to \bar E \in I$,
  and, 
  by Theorem \ref{thm:convergence}, $\rho_k \to \rho$ which
  solves the MVP in $\Lambda$ with energy $\bar E\in I$.
  Since the solution of the MVP 
  with energy $\bar E$ corresponds only to   
  $U=U^0(\bar E)$ or $U=U^1(\bar E)$, we get a contradiction.
  
  Now we  prove that for sufficiently large $n$, the functions
  in $\Cal M_n(E)$ are the solutions of MFE given by 
  Proposition \ref{prop:gluing}.
  Namely, choosing $r<r_I$, for $n\ge n_{r}$ as in Corollary \ref{corollario-E},
  for any $E$ in $I$ there exists a unique solution
  $U^{n,\sigma}(E)$   of the MFE on $\Omega_n$
  with energy $E$ in $B_r(U^\sigma(E))$.
  Therefore
  $$\Cal M_n(E)\subset \{U^{n,0}(E), U^{n,1}(E)\}.$$
  The existence of the (unique) value $E_*^n$ for which
  MVP has two solutions easily follow form the fact that,
  as proved in Lemma \ref{lemma:inv},
  the entropy associated to $U^{n,\sigma}(E)$ is a $C^2$ function in
  $E$ and converges uniformly with its derivative to
  the entropy of the $U^{\sigma}(E)$.
\end{proof}

\appendix
\section{Deformed circles}\label{appendice:A}

We prove eq.s \eqref{eq:ezeta}, recalling that
that the function $\mf e(\mu)$ is invariant
for scaling, while $\mf z(\mu)$ is linear in the area.
So we have to consider  the deformation $\Lambda_\eta$
of the  circle $D$ of radius $1$, 
with the conformal map
$z \to z + \eps z^3$, with $\eps^2 =\eta$.
The determinant of the Jacobian of this map is
$J_\eps=1+a\eps + b \eps^2$, with $a = 6(x^2-y^2)$,
and $b= 9(x^2+y^2)^2$.

In order to find $\mf e(\mu)$ and $\mf z(\mu)$,
we first calculate the free energy,
noting that, for $\beta<0$, 
$$\ms F(\beta) = \sup_{\rho} \Cal G(\beta, \rho),$$
where
$$
\Cal G(\beta, \rho)  = 
-\Cal E(\rho)-\frac{1}{\beta}\log
\Cal Z(\beta,\rho), \ \ \Cal Z(\beta,\rho) =
\int_{\Lambda_\eta} \e^{-\beta \Psi},
$$
(see Section 8 of \cite{CLMP2}).
In fact the Euler-Lagrange equation for the 
variational principle is the MFE equation,
and, if $\rho$ is a solution of the MFE equation then
$$\begin{aligned}
  \Cal G(\beta,\rho) &=  - \Cal E(\rho) - \frac 1\beta
\log \Cal Z(\beta,\rho) =
-\Cal E(\rho) -\frac 1\beta (\Cal S(\rho) - 2\beta \Cal E(\rho))
= \Cal E(\rho) -\frac 1\beta
\Cal S(\rho) \\
&= \Cal F(\beta,\rho)=\ms F(\beta).\end{aligned}$$
The advantage of this formulation is that
the expression of the energy is invariant for conformal transformation,
then we maximize $\Cal G$ for $\rho$ supported on $\Lambda_\eta$
if we maximize in $\Phi$, defined on $D$, the functional 
$$\Cal I_\eps(\beta,\Phi) = -\frac 12 \int_D |\grad \Phi|^2
-\frac 1\beta \log \int_D J_\eps \e^{-\beta \Phi}.$$
The corresponding
Euler-Lagrange equation is
$$-\lap \Phi_\eps = \frac {J_\eps}{Z_\eps} \e^{-\beta \Phi_\eps},\ \
Z_\eps = \int_D J_\eps \e^{-\beta \Phi_\eps},$$
with homogeneous Dirichlet boundary condition on $\pa D$.
For $\eps=0$ we clearly get the solution of the MFE on $D$,
that 
we call $\Phi$, with associated density $\rho$, and normalization
$Z$. 
We denote
$$\Phi' = \left. \opde{\eps}\right|_{\eps=0} \Phi_\eps,
\ \ \Phi'' = \left.
  \dfrac {\mathrm{d}^2 \phantom{\eps}}{\mathrm{d}\eps^2}
  \right|_{\eps=0} \Phi_\eps.
$$
The equation for $\Phi'$ is 
\begin{equation}
  \label{eq:phiprimo}
  -\lap \Phi' = \rho (a-\beta \Phi') -\rho C,
\end{equation}
where
$$C = \int_D \rho (a-\beta \Phi') = -\beta \int_D \rho \Phi',$$
since $\rho$ is radial and $a=6r^2\cos (2\vartheta)$ in polar coordinates.

Then, by eq. \eqref{eq:phiprimo} we have
$$
\left. \opde{\eps}\right|_{\eps=0} \Cal I_\eps(\beta,\Psi_\eps)
= \int_D \left( \lap \Phi \, \Phi' + \rho \Phi' + \frac 1\beta \rho a
\right)=0.$$
%
The second derivative in $\eps=0$ if given by
$$\begin{aligned}
\Cal I'' &\coloneqq  \left. \dfrac{ \de^2 \phantom{\eps^2}}{
    \de \eps^2}
\right|_{\eps=0} \Cal I_\eps(\beta,\Psi_\eps) \\
&= -\int_D |\grad \Phi'|^2 + \int_D \lap \Phi \Phi''
- \frac 1\beta
\int_D \rho (2b - \beta \Phi'' - \beta a \Phi' -\beta \Phi'(a-\beta \Phi'))
+\frac 1\beta C^2.
\end{aligned}$$
The terms in $\Phi''$ vanishes, since $\Phi$ solves the MFE.
Moreover
$$-\int_D|\grad \Phi'|^2 = \int_D \lap \Phi' \Phi' =
-\int_D \rho \Phi' (a-\beta \Phi') + C \int_D \rho \Phi',$$
and $C  = -\beta \int_D \rho \Phi'$. Therefore
\begin{equation}
  \label{eq:G2}
  \Cal I'' = -\frac 1\beta \int_D \rho (2b - \beta a\Phi') +
  \gamma \int_D \rho \Phi' + 
  \frac 1\beta \gamma^2 = -\frac 1\beta \int_D \rho
  (2b - \beta a \Phi').
\end{equation}

In order to compute $\Cal I''$ we have do find $\Phi'$.
By noticing that $a$ is harmonic, 
we can solve the equation for $\Phi'$ by searching for a solution
of the form
$$\Phi' = \frac a\beta  + \xi(r)  \cos (2\vartheta),$$
where $(r,\vartheta)$ are  polar coordinates, 
with the boundary condition $\xi(1) = -\frac 6\beta=
\frac 3{4\pi \mu}$.
We 
set $p=\mu/(1-\mu)$, and we note that
$$\rho = \frac 1{\pi (1-\mu)} \frac 1{(1+pr^2)^2}.$$
We search for $\xi(r)=c \alpha (pr^2) $.
The equation for  $\alpha$ in the variable 
$s=pr^2$   
is 
$$(1+s)^2 ( (s \pa_s)^2 \alpha - \alpha) = -2 s \alpha.$$
There exists only one solution,  bounded in $0$ and with $\alpha(1)=1$,
given by
$$\alpha(s) = \frac s{1+s} + \frac s2.$$
Then
$$\xi(r) = \frac 3{4\pi \mu \alpha(p)} \alpha(pr^2).$$
Computing $\Phi'$ and inserting its expression in \eqref{eq:G2}
we finally get
$$\ms F_{\Lambda_\eta}(\beta) = \ms F (\beta) - \eta \frac 1\beta g(\beta) +
o(\eta)$$
where
$\ms F(\beta)$ is the free energy for $D$, and the correction
$g$ is given by 
$$ g(\beta)  = 6\frac {1-\mu}{1-\frac 23 \mu}$$
(we scaled $\eta$ with the constant 6, for simplicity of notation).

Now we have to find the relation between $E$, $\beta$, $S$
We first define 
$f_\eta (\beta) = \beta  \ms F_\eta (\beta) = f(\beta) - \eta g(\beta)
+o(\eta)$.
The entropy is given by
$S_\eta (E) = \inf_{\beta} (\beta E - f_\eta(\beta))$,
where the infimum is attained in 
$E_\eta(\beta) = f'_\eta(\beta) = f'(\beta) - \eta g'(\beta)+o(\eta)$,
from which we get the expression of $\mf e_\eta(\mu)$
in \eqref{eq:ezeta}, where $\tau=\frac 1{8\pi} \pa_\beta g = -\pa_\mu g$.

The entropy in function of $\beta$ is  
$S_\eta (\beta) = S(\beta) + \eta (g - \beta \pa_\beta g)+o(\eta)$
where $S(\beta)$ is the value for $\eta = 0$, i.e. the case
of the disk.
Denoting by $Z(\beta)$ the normalization factor for the disk,
we have 
$$\log Z_\eta(\beta) = S_\eta(\beta) - 2 \beta E_\eta(\beta)
= \log Z + \eta ( g + \beta
\pa_\beta g)+o(\eta).$$
Since $\beta \pa_\beta g = \mu \pa_\mu g$,
we obtain the expression of 
$z_{1,\eta}$ in \eqref{eq:ezeta},  where  $\zeta(\mu)  = g + \mu \pa_\mu g$.

\end{document}